%% file: main.tex
\documentclass[acmsmall,nonacm]{acmart}

\usepackage{threeparttable}
\usepackage{multicol,multirow}
\usepackage{makecell}
\usepackage{wrapfig}
\usepackage{booktabs}
\usepackage{tabularx}
\usepackage{bm}
\usepackage{setspace}
\usepackage{algorithm}
\usepackage{tikz}
\usepackage{algpseudocode}
\usepackage{algorithmicx}
\newtheorem{theorem}{Theorem}
\usepackage{hyperref}
\usepackage{cleveref}
\usepackage{wasysym}

\newcommand*\circledblack[1]{\tikz[baseline=(char.base)]{\node[shape=circle,fill=black,text=white,inner sep=0.6pt] (char) {#1};}}

\usepackage{pifont}
\usepackage{xcolor}
\definecolor{myyellow}{HTML}{FFC519}
\definecolor{myblue}{HTML}{79BEEA}
\definecolor{mygreen}{RGB}{0, 128, 0}
\definecolor{myred}{RGB}{192, 0, 0}

\usepackage{draftwatermark}
\SetWatermarkText{Accepted by PLDI 2026}
\SetWatermarkScale{0.45} 
\SetWatermarkColor[gray]{0.90} 

\AtBeginDocument{%
  }

\setcopyright{acmlicensed}
\copyrightyear{2018}
\acmYear{2018}
\acmDOI{XXXXXXX.XXXXXXX}
\acmConference[Conference acronym 'XX]{Make sure to enter the correct
  conference title from your rights confirmation email}{June 03--05,
  2018}{Woodstock, NY}
\acmISBN{978-1-4503-XXXX-X/2018/06}

\usepackage[colorinlistoftodos,prependcaption]{todonotes}

\begin{document}

\title{NEURA: A Unified and Retargetable Compilation Framework for Coarse-Grained Reconfigurable Architectures}

\author{Shangkun Li}
\orcid{0009-0000-7623-4178}
\email{shangkun.li@connect.ust.hk}
\affiliation{%
  \institution{The Hong Kong University of Science and Technology}
  \country{Hong Kong}
}

\author{Jinming Ge}
\orcid{0009-0008-7986-4886}
\email{jgeab@connect.ust.hk}
\affiliation{%
  \institution{The Hong Kong University of Science and Technology}
  \country{Hong Kong}
}

\author{Diyuan Tao}
\orcid{0009-0008-2696-6651}
\email{1173710114az@gmail.com}
\affiliation{%
  \institution{Independent Researcher}
  \country{China}}

\author{Zeyu Li}
\orcid{0009-0003-5104-7466}
\email{zliki@connect.ust.hk}
\affiliation{%
  \institution{The Hong Kong University of Science and Technology}
  \country{Hong Kong}
}

\author{Jiawei Liang}
\email{jliangbr@connect.ust.hk}
\affiliation{%
  \institution{The Hong Kong University of Science and Technology}
  \country{Hong Kong}
}

\author{Linfeng Du}
\orcid{0000-0002-3007-4890}
\email{linfeng.du@connect.ust.hk}
\affiliation{%
  \institution{The Hong Kong University of Science and Technology}
  \country{Hong Kong}
}

\author{Jiang Xu}
\orcid{0000-0001-9089-7752}
\email{jiang.xu@hkust-gz.edu.cn}
\affiliation{%
  \institution{The Hong Kong University of Science and Technology (Guangzhou)}
  \country{China}
}

\author{Wei Zhang}
\authornote{Co-Corresponding Authors}
\orcid{0000-0002-7622-6714}
\email{wei.zhang@ust.hk}
\affiliation{%
  \institution{The Hong Kong University of Science and Technology}
  \country{Hong Kong}
}

\author{Cheng Tan}
\orcid{0000-0003-3727-2889}
\email{chengtan@google.com}
\authornotemark[1]
\affiliation{%
  \institution{Google and Arizona State University}
  \country{USA}
}


\begin{abstract}
Coarse-Grained Reconfigurable Architectures (CGRAs) are a promising and versatile accelerator platform, offering a balance between the performance and efficiency of specialized accelerators and the software programmability. However, their full potential is severely hindered by control flow in accelerated kernels, as the control flow (e.g., loops, branches) is fundamentally incompatible with the parallel, data-driven CGRA fabric. Prior strategies to resolve this mismatch in CGRA kernel acceleration are either inefficient, sacrificing performance for generality, or lack generality due to the difficulty of adapting them across different execution models. Thus, a general and unified solution for efficient CGRA kernel acceleration remains elusive.

This paper introduces NEURA, a unified and retargetable compilation framework that systematically resolves the control-dataflow mismatch in CGRAs. NEURA's core innovation is a novel, pure dataflow intermediate representation (IR) built on a predicated type system. In this IR, control contexts are embedded as a predicate within each data, making control an intrinsic property of data. This mechanism enables NEURA to systematically flatten complex control flow into a single unified dataflow graph. This unified representation decouples kernel representation from hardware, empowering NEURA to retarget diverse CGRAs with different execution models and microarchitectural features. When targeted to a high-performance \textbf{\textit{spatio-temporal}} CGRA, NEURA delivers a $2.20\times$ speedup on kernel benchmarks and up to $2.71\times$ geometric mean speedup on real-world applications over state-of-the-art (SOTA) high-performance baselines. It also provides a competitive solution against the SOTA low-power CGRA when retargeted to a \textbf{\textit{spatial-only}} CGRA. \textbf{NEURA is open-source and available at \texttt{\href{https://github.com/coredac/neura}{https://github.com/coredac/neura}}.}
\end{abstract}




\maketitle

\input{Ch1-Introduction}

\input{Ch2-Background}

\input{Ch3-Overview}

\input{Ch4-IR}

\input{Ch5-Transform}

\input{Ch6-Optimization}

\input{Ch7-Implementation}

\input{Ch8-Evaluation}

\input{Ch9-Related}

\section{Conclusion}\label{sec10:conclusion}
This paper presents NEURA, a unified and retargetable compilation framework for CGRAs that systematically resolves the \textbf{\textit{control-dataflow semantic gap}}. NEURA's core innovation is a novel pure dataflow IR that uses a predicated type system to embed hierarchical control context into data values. This enables flattening a kernel with complex control flow into a single, unified DFG, eliminating the serialization bottlenecks and rigid mapping limitations of prior approaches. Our evaluation validates NEURA's effectiveness and retargetability with negligible hardware overhead. When targeted to a high-performance \textbf{\textit{spatio-temporal}} architecture, NEURA achieves a $2.20\times$ geomean speedup on kernels and up to $2.71\times$ geomean speedup on real-world applications over leading baselines. When retargeted to a low-power \textbf{\textit{spatial-only}} architecture, it provides a competitive solution against the SOTA low-power framework.

\bibliographystyle{ACM-Reference-Format}
\bibliography{ref}










\end{document}

%% file: Ch1-Introduction.tex
\section{Introduction}\label{sec1:introduction}

The relentless pursuit of performance in the post-Moore era \cite{dark-silicon-ISCA2011, dark-silicon-server-Micro2011, DSA-accelerators-CommunACM2020} has spurred specialized hardware accelerators. However, these accelerators are often constrained by lengthy design cycles and rigid architectures \cite{HiPER-ISCA2025, Dadu-RBD-MICRO2023, dadiannao-MICRO2014, Cambricon-X-MICRO2016}, struggling to adapt to rapid algorithmic evolution. Furthermore, edge devices cannot deploy multiple specialized accelerators due to power and area constraints \cite{AcceleratorWall-HPCA2019}. These limitations have spurred significant interest in Coarse-Grained Reconfigurable Architectures (CGRAs), which offer a balanced solution between the performance and efficiency of specialized accelerators and the flexibility for diverse computational demands.

A CGRA is a spatial array of interconnected compute tiles \cite{REVEL-HPCA2020, Amber-JSSC2024, ICED-MICRO2024, NUPEA-ISCA2025, Ripple-PLDI2025, RipTide-MICRO2022, Plasticine-ISCA2017, UE-CGRA-HPCA2021}, communicating via a Network-on-Chip (NoC), as shown in Fig. \ref{fig:rda}(a). Each tile contains several function units (FUs) for computation, a crossbar for intra- or inter-tile routing, register files for temporary data storage, and a control memory, as shown in Fig. \ref{fig:rda}(b). By loading configurations into control memory, the behavior of the FUs and the connections established by the crossbar can be reconfigured cycle-by-cycle. A CGRA compiler typically transforms a computational kernel into a Control-Data Flow Graph (CDFG) representation \cite{Hyperblock-MICRO1992, DFG-CACM1976}, as this naturally captures both computation and control. The CDFG consists of a Control Flow Graph (CFG) and Data Flow Graphs (DFGs). A DFG depicts operations as nodes and data dependencies as edges, as shown in Fig. \ref{fig:execution-model}(b). Each DFG corresponds to the computation within a basic block (BB). The CFG is a graph whose nodes are BBs, and its edges represent the control flows between BBs, as shown in Fig. \ref{fig:execution-model}(a). The compiler maps the DFG nodes onto tiles' FUs and DFG edges onto the routing resources (e.g., crossbar, registers). This allows the CGRA to exploit instruction-level parallelism (ILP) within a single DFG.

\begin{wrapfigure}{r}{0.45\textwidth}
  \vspace*{-0.8\baselineskip}
  \centering
  \includegraphics[width=0.45\textwidth]{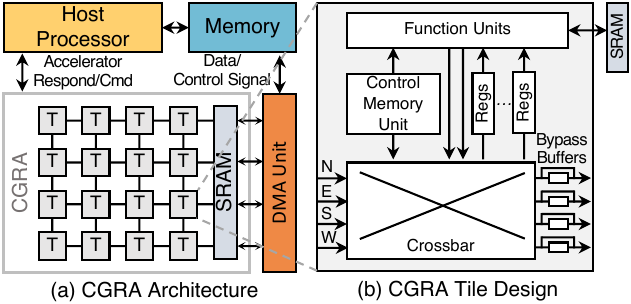}
  \vspace*{-1.2\baselineskip}
  \caption{CGRA Architecture --- (a) The CGRA is invoked via the accelerator command interface by the host processor. Control signals generated by the host processor and data required for computation are loaded through the Direct Memory Access (DMA) Unit from shared memory into the CGRA's SRAM and each tile's control memory unit. (b) CGRA tile components.}
  \label{fig:rda}
  \vspace*{-0.8\baselineskip}
\end{wrapfigure}

However, while the DFG maps naturally onto the CGRA, the CFG presents a fundamental challenge to kernel acceleration. The CFG is the sequential, control-driven logic (e.g., branches, loops) that dictates transitions between different DFGs. A CGRA, in contrast, is configured to execute a specific DFG in a parallel, data-driven manner. It lacks a native mechanism to directly interpret or execute these CFG constructs. This creates a significant \textbf{\textit{control-dataflow semantic gap}} when accelerating a kernel on CGRA: \textbf{\textit{a fundamental mismatch between the sequential control logic of the CFG and the parallel dataflow execution of CGRAs.}} This gap severely constrains CGRA's ability to exploit ILP across DFGs, becoming a primary obstacle to achieving the full performance potential of CGRAs.

Prior attempts to bridge this \textbf{\textit{control-dataflow semantic gap}} have followed several strategies, but our observations indicate that they all suffer from critical pitfalls: (1) the common approach of managing the CFG externally using host CPUs or dedicated hardware controllers inherently serializes kernel execution at the BB \cite{Marionette-MICRO2023, Plasticine-ISCA2017} or larger fused-block level \cite{SARA-ISCA2021, TRIPS-2004TACO}, creating performance bottlenecks due to communication and reconfiguration overheads, (2) the strategy of flattening the entire CDFG into a single DFG using steering control \cite{Ripple-PLDI2025, RipTide-MICRO2022, NUPEA-ISCA2025} is widely adopted for \textbf{\textit{spatial-only}} execution (see Sec. \ref{subsec2.1:exemodel}) but difficult to adapt to \textbf{\textit{spatio-temporal}} execution, thereby limiting its architectural generality, and (3) existing predication techniques \cite{IfConversion-1983POPL, Hyperblock-MICRO1992, LLVMIfConversion-2026online}, applied in some CGRA compilers \cite{OpenCGRA-ICCD2020, Morpher-2022WOSET, BranchAware-2014DAC}, are effective for simple branches (e.g., if-conversions) or a single loop but fail to represent hierarchical control dependencies (e.g., nested loop controls), thus still relying on serialized execution for the remaining control structures \cite{SARA-ISCA2021, TRIPS-2004TACO}. The pitfalls inherent in these strategies reveal that existing approaches fail to provide a unified, general, and retargetable approach capable of efficiently managing control flows across diverse CGRA architectures.

To address the pitfalls of prior works, we propose \textbf{NEURA} --- a novel compilation framework centered on \textbf{\underline{N}}atively \textbf{\underline{E}}xecuting a \textbf{\underline{U}}nified Intermediate \textbf{\underline{R}}epresentation on CGR\textbf{\underline{A}}s. NEURA's core idea is a new intermediate representation (IR) designed specifically for CGRA dataflow execution. This IR is designed to (1) \textbf{express control flows of a kernel in a unified dataflow manner for high ILP}, (2) \textbf{serve as an extensible, retargetable foundation for diverse CGRAs by decoupling from specific execution models and hardware primitives}, and (3) \textbf{provide native support for hierarchical control dependencies}. By seamlessly integrating control flows into data dependencies, this unified representation eliminates the \textbf{\textit{control-dataflow semantic gap}}, enabling efficient, holistic kernel execution on CGRAs. The specific contributions are as follows:

\begin{itemize}
    \item \textbf{A Unified and Extensible Dataflow IR:} We propose the \textbf{NEURA Dataflow IR}, a pure dataflow representation built on a novel predicated type system and a general operation set. It holistically captures complex, hierarchical control structures in kernels. Its extensible design allows the IR to model diverse microarchitectural features and execution models;
    \item \textbf{A Systematic Methodology for Flattening Control Flow:} We introduce a systematic and provably complete methodology that automatically transforms kernels from a conventional CDFG representation into the pure NEURA Dataflow IR;
    \item \textbf{A Versatile and Retargetable Compilation Framework:} We propose an end-to-end CGRA compilation framework, integrating frontends, transformations, hardware-agnostic/-specific optimizations, mapping, an IR interpreter, and a cycle-accurate simulator. Leveraging the IR's extensibility, it provides a retargetable solution for diverse CGRA architectures;
    \item \textbf{A Comprehensive Evaluation on NEURA:} We conduct extensive evaluations validating NEURA as a unified, retargetable framework for both high-performance and low-power domains across diverse benchmarks. When targeting high-performance \textbf{\textit{spatio-temporal}} architecture, NEURA achieves state-of-the-art (SOTA) performance, delivering an average speedup of $2.20\times$ on kernel benchmarks and up to $2.71\times$ geometric mean speedup on real-world applications over leading baselines. NEURA also proves its generality by offering a competitive solution when retargeted to a low-power \textbf{\textit{spatial-only}} architecture in a direct comparison against the SOTA low-power framework.
\end{itemize}

%% file: Ch2-Background.tex
\section{Background and Motivation}\label{sec2:background}

This section begins by introducing the two primary execution models of CGRAs and establishing the trade-offs that motivate our design to support both execution models in Sec. \ref{subsec2.1:exemodel}. We then use a motivating example to deconstruct the three prevailing strategies for handling control flow in Sec. \ref{subsec2.2:controlflow}, revealing the pitfalls that necessitate the design of NEURA.

\subsection{Execution Models of CGRAs}\label{subsec2.1:exemodel}
To understand the compilation challenges for CGRAs, we first consider their execution models. Using the synthetic kernel in Fig. \ref{fig:execution-model}(a), we illustrate the execution of a simple loop body's (ignoring loop control for simplicity) DFG. CGRAs typically employ modulo scheduling \cite{ModuloScheduling-MICRO1994} to pipeline loop iterations, initiating a new iteration every Initiation Interval (II) cycles to maximize ILP. How the DFG is mapped to the hardware defines two distinct execution models with different trade-offs.

\textit{\textbf{Spatial-Only}} execution model statically maps a DFG or a partitioned subgraph onto the CGRA's tile array, creating a one-to-one mapping between DFG nodes and physical tiles, as shown in the left part of Fig. \ref{fig:execution-model}(c). Once configured, each tile's function and routing are immutable for the duration of the mapped graph's execution. This static, fully-spatial approach eliminates the cycle-by-cycle runtime reconfiguration overhead, making it exceptionally energy-efficient. However, this static nature introduces inflexibility when executing a DFG larger than the tile array. The compiler must partition the large DFG into multiple subgraphs and sequence their execution. This introduces compiler challenges for optimal DFG partitioning and incurs context-switching overheads when reconfiguring the array between subgraphs. Due to these trade-offs, this execution model is an ideal choice for \textbf{low-power domains} \cite{RipTide-MICRO2022, SNAFU-ISCA2021}.

\textit{\textbf{Spatio-Temporal}} execution model combines cycle-by-cycle resource-multiplexing with spatial mapping, as shown in the right part of Fig. \ref{fig:execution-model}(c). It allows each tile to be reconfigured every cycle to execute different operations \cite{OpenCGRA-ICCD2020, CGRA-ME-2017ASAP, DRESC-FPT2002}. This dynamic reconfigurability provides high flexibility, enabling even a compact tile array to execute large DFGs by dynamically reusing resources across cycles. This leads to higher hardware utilization and generality across diverse workloads. However, this execution model relies on complex \textbf{\textit{spatio-temporal}} scheduling. Furthermore, cycle-by-cycle reconfiguration incurs high energy and area overheads. Due to these trade-offs, this execution model is typically the preferred choice for \textbf{general-purpose and high-performance CGRAs}.

\begin{wrapfigure}{r}{0.63\textwidth}
  \vspace*{-0.6\baselineskip}
  \centering
  \includegraphics[width=0.63\textwidth]{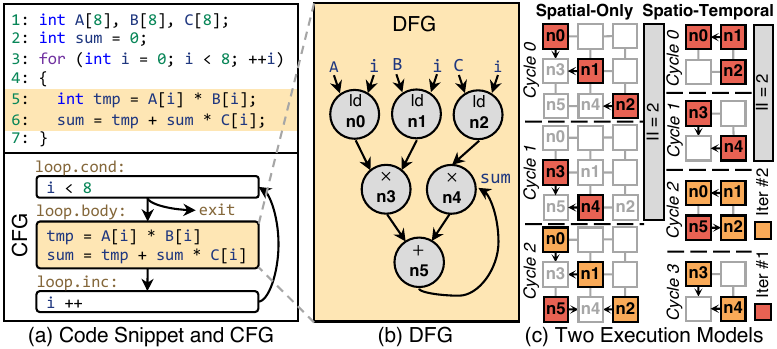}
  \vspace*{-1.2\baselineskip}
  \caption{Kernel Representation and Execution Models --- (a) A synthetic kernel and the CFG of this kernel. We show the source code instead of the assembly code in the CFG for simplicity. (b) The DFG of the loop body. (c) The 6-node DFG requires a larger $3\times3$ array to execute in a spatial-only CGRA. The spatio-temporal CGRA can execute the same DFG on a smaller $2\times2$ array by time-multiplexing the tiles.}
  \label{fig:execution-model}
  \vspace*{-0.6\baselineskip}
\end{wrapfigure}
The salient characteristics and trade-offs of these two execution models are summarized in Table \ref{table:exec-model}. As discussed, both models possess unique strengths suited for distinct requirements. Furthermore, both remain prevalent and vital to the CGRA ecosystem, as evidenced by the continuous emergence of recent \textbf{\textit{spatial-only}} \cite{Ripple-PLDI2025,RipTide-MICRO2022,NUPEA-ISCA2025,Pipestitch-MICRO2023} and \textbf{\textit{spatio-temporal}} \cite{PICACHU-ASPLOS2025, FexMo-MICRO2025, Plaid-ASPLOS2025, ICED-MICRO2024} architectures. However, existing compiler frameworks often rely on underlying abstractions \cite{Ripple-PLDI2025, NUPEA-ISCA2025, ICED-MICRO2024, SARA-ISCA2021, Marionette-MICRO2023} that are difficult to generalize across both execution models. They lack a versatile dataflow abstraction that can decouple computational and control logic from the specific execution model. Since both models are indispensable, NEURA's objective is to provide a unified, retargetable compilation framework that seamlessly supports both.

\begin{table}[t]
\centering
\caption{Execution Models of CGRAs.}
\label{table:exec-model}
\vspace*{-0.7\baselineskip}
\resizebox{\linewidth}{!}{ 
\begin{tabular}{lccccc}
\toprule

\textbf{\multirow{2}{*}{\makecell[c]{CGRA\\Execution Model}}} & \textbf{\multirow{2}{*}{\makecell[c]{Generality}}} & \textbf{\multirow{2}{*}{\makecell[c]{Performance}}} & \textbf{\multirow{2}{*}{\makecell[c]{Energy\\Efficiency}}} & \textbf{\multirow{2}{*}{\makecell[c]{Hardware\\Utilization}}} & \textbf{\multirow{2}{*}{\makecell[c]{Representative\\Works}}}\\\\
\toprule

\multirow{2}{*}{\textbf{\makecell[c]{Spatial-Only}}} & \multirow{2}{*}{\makecell[c]{Low}} & \multirow{2}{*}{\makecell[c]{Medium$^\mathrm{\ast}$}} & \multirow{2}{*}{\makecell[c]{High}} & \multirow{2}{*}{\makecell[c]{Medium$^\mathrm{\ast}$}} & \multirow{2}{*}{\makecell[c]{SNAFU \cite{SNAFU-ISCA2021}, RipTide \cite{RipTide-MICRO2022}, Marionette \cite{Marionette-MICRO2023}, \\ NUPEA \cite{NUPEA-ISCA2025}, Plasticine \cite{Plasticine-ISCA2017}, DySER \cite{DySER-IEEEMicro2012}}}  \\ \\
\hline
\multirow{2}{*}{\textbf{\makecell[c]{Spatio-Temporal}}} & \multirow{2}{*}{\makecell[c]{High}} & \multirow{2}{*}{\makecell[c]{High}} & \multirow{2}{*}{\makecell[c]{Low}} & \multirow{2}{*}{\makecell[c]{High}} & \multirow{2}{*}{\makecell[c]{ICED \cite{ICED-MICRO2024}, Plaid \cite{Plaid-ASPLOS2025}, DRIPS \cite{DRIPS-HPCA2022}, \\ HyCUBE \cite{HyCUBE-DAC2017}, ADRES \cite{ADRES-FPL2003}}} \\ \\

\bottomrule

\multicolumn{6}{l}{$\mathrm{\ast}$ Performance and hardware utilization degrade when executing irregular dataflow patterns.}
\end{tabular}
}
\vspace*{-0.8\baselineskip}
\end{table}

\subsection{Pitfalls of Existing Strategies to Bridge the Control-Dataflow Semantic Gap}\label{subsec2.2:controlflow}

As introduced in Sec. \ref{sec1:introduction}, accelerating kernels on CGRAs involves bridging the \textbf{\textit{control-dataflow semantic gap}} between a kernel's CDFG representation and the hardware's dataflow fabric. Three primary strategies have emerged to tackle this challenge, but each introduces its own pitfall, ultimately failing to unlock the full potential of CGRAs. The motivating example in Fig. \ref{fig:motivation}(a) --- a kernel with a three-level control structure: two-level imperfect nested loops and an inner \texttt{if-else} branch --- exposes the fundamental limitations of existing strategies.

\begin{figure}[t]
    \centering
    \includegraphics[width=1.0\textwidth]{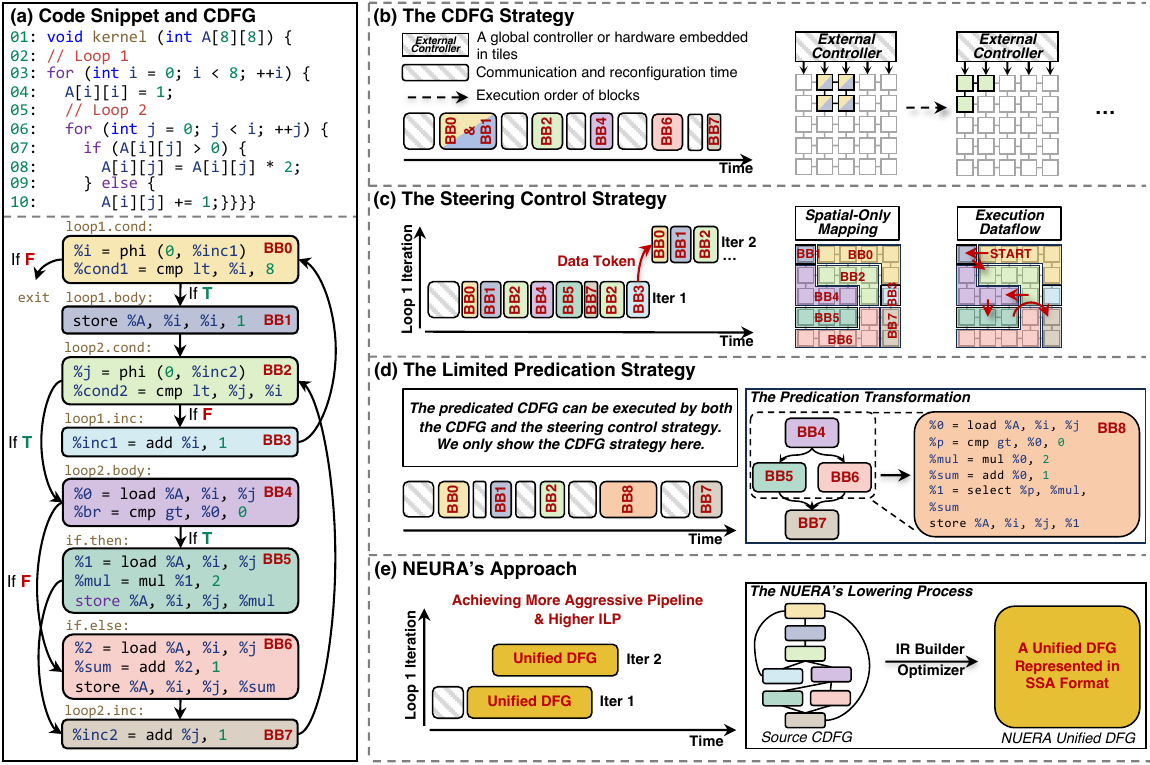}
    \vspace*{-1.0\baselineskip}
    \caption{Motivating Example --- (a) A synthetic kernel with imperfect nested loop and branch divergence and its CDFG. (b) The CDFG strategy serializes the execution of each BB via an external controller. (c) The steering control strategy flattens the control flow in the CFG for spatial-only execution. (d) Limited predication transforms the branch divergence into a single BB but fails to resolve the nested control flow with loops. (e) NEURA represents the kernel as a unified DFG, exploiting both intra- and inter-BB parallelism.
    }
    \label{fig:motivation}
    \vspace*{-1.0\baselineskip}
\end{figure}

\subsubsection{The CDFG Strategy: Serialization via External Controller} This strategy offloads CFG management to external controllers, such as a host CPU or dedicated control hardware \cite{Marionette-MICRO2023, SARA-ISCA2021, Plasticine-ISCA2017, DySER-IEEEMicro2012, TRIPS-2004TACO}. As shown in Fig. \ref {fig:motivation}(b), the controller sequentially dispatches each BB or larger fused-block to the CGRA. For each BB or fused-block, the controller first configures the CGRA fabric for the corresponding DFG, triggers its execution, and then waits for completion before proceeding to the next (fused-)block. This stall introduces a lock-step interaction, creating execution bubbles due to communication and CGRA reconfiguration. This results in \textit{\textbf{Pitfall I: Serializing execution at the BB or fused-block level} severely undermines inter-block parallelism.}

\subsubsection{The Steering Control Strategy: Flattening the CDFG via Steering Control} This strategy flattens the CDFG into a single DFG by converting control dependencies into data dependencies using steering control ($\phi^{-1}$) \cite{SteerControl-PL1991, DataflowProcessor-ISCA1974, Ripple-PLDI2025, RipTide-MICRO2022}. At runtime, data values are steered through pre-configured paths based on condition results, as shown in Fig. \ref{fig:motivation}(c). While this creates a unified DFG, it is predominantly adopted in the \textbf{\textit{spatial-only}} execution model. Because \textbf{\textit{spatio-temporal}} execution requires the compiler to schedule operations across both spatial and temporal dimensions, it is difficult to determine the execution timing and placement for operations that depend on dynamically steered data tokens, since these tokens may experience variable latencies when traversing different execution paths. Consequently, this strategy is hard to adapt to \textbf{\textit{spatio-temporal}} execution models. Furthermore, statically encoding complex control paths as physical routes often results in long, circuitous routing paths that drastically degrade performance. These drawbacks constitute \textit{\textbf{Pitfall II: Steering control is difficult to adapt to the \textbf{\textit{spatio-temporal}} execution model, sacrificing performance and architectural generality.}}

\subsubsection{The Limited Predication Strategy: Failure in Hierarchical Predicates} Conventional predication converts control dependencies into data dependencies \cite{IfConversion-1983POPL, Hyperblock-MICRO1992, LLVMIfConversion-2026online}, enabling concurrent branch execution across both \textbf{\textit{spatial-only}} and \textbf{\textit{spatio-temporal}} execution models. However, existing predicate-based CGRA compilers \cite{BranchAware-2014DAC, OpenCGRA-ICCD2020, ICED-MICRO2024, Plaid-ASPLOS2025} are limited to simple if-conversions or single loops, failing to handle hierarchical predicates in nested control (e.g., nested loops, branches in loops). As shown in Fig. \ref{fig:motivation}(a), an instruction in BB5 executes only if hierarchical predicates are met -- both outer (Loop 1) and inner (Loop 2) loops are active, and the branch condition is true. These predicates are computed in different BBs, and conventional CDFG-based IRs (e.g., LLVM \cite{LLVM-CGO2004} IR) lack the semantics to propagate the loop predicates (e.g., \texttt{\%cond1} and \texttt{\%cond2} from BB0 and BB2) into an inner BB (e.g., BB5) and merge them with the block's internal logic. Consequently, as shown in Fig. \ref{fig:motivation}(d), the compiler can only apply predication locally to branch divergences (e.g., BB5 and BB6), leaving the surrounding loop structure stranded in CFG. To manage these stranded structures, existing works (e.g., TRIPS \cite{TRIPS-2004TACO}) are forced to adopt a hybrid execution: utilizing predication for if-conversions and relying on the CDFG strategy of \textbf{\textit{Pitfall I}} for loop structures. However, this hybrid execution inherently restricts the ability to fully exploit inter-block parallelism (e.g., loop iteration pipelining). This representational flaw constitutes \textit{\textbf{Pitfall III: Failing to unify hierarchical predicates forces a retreat to the serialized CDFG execution strategy.}}

A comparison between NEURA and other CGRA frameworks is listed in Table \ref{table:comparison}. This table highlights NEURA's unique capabilities, which are enabled by its novel dataflow IR. As shown in Fig. \ref{fig:motivation}(e), NEURA's IR flattens the entire kernel into a single, unified DFG. This holistic representation eliminates the serialization bubbles and reconfiguration overhead of the CDFG and limited predication strategies, allowing NEURA to fully exploit both intra- and inter-BB ILP. Furthermore, this unified IR is decoupled from any single execution model, thus avoiding the rigid mapping and scalability limitations of the steering control strategy.

%% file: Ch3-Overview.tex
\section{NEURA Overview}\label{sec3:overview}

\begin{table}[t]
\centering
\caption{Comparison between NEURA and Existing CGRA Frameworks (\CIRCLE -- Fully Supported, \LEFTcircle -- Partially Supported, \Circle -- Not Supported)
}
\label{table:comparison}
\vspace*{-0.7\baselineskip}
\resizebox{\linewidth}{!}{ 
\begin{tabular}{lcccccccc}
\toprule

\textbf{\multirow{2}{*}{\makecell[c]{Framework}}} & \textbf{\multirow{2}{*}{\makecell[c]{Input\\Language}}} & \textbf{\multirow{2}{*}{\makecell[c]{Control Flow\\Strategy}}} & \textbf{\multirow{2}{*}{\makecell[c]{Multi-\\-Frontend}}} & \textbf{\multirow{2}{*}{\makecell[c]{Microarchitectural\\Adaptability$^\mathrm{\ast}$}}} & \textbf{\multirow{2}{*}{\makecell[c]{Holistic Dataflow\\Representation$^\mathrm{\dagger}$}}} & \textbf{\multirow{2}{*}{\makecell[c]{Execution Model\\Versatility$^\mathrm{\ddagger}$}}} & \multicolumn{2}{c}{\textbf{Optimization}}\\
\cmidrule[\heavyrulewidth](lr){8-9}
 & & & & & & & \textbf{HW-Agnostic} & \textbf{HW-Specific}
\\
\toprule

Marionette \cite{Marionette-MICRO2023} & C/C++ & CDFG & \LEFTcircle & \Circle & \Circle & \Circle & \Circle & \LEFTcircle \\

Plasticine \cite{Plasticine-ISCA2017} & Spatial \cite{Spatial-PLDI2018} & CDFG & \Circle & \Circle & \Circle & \Circle & \Circle & \LEFTcircle \\

Ripple \cite{Ripple-PLDI2025} & C/C++ & Steering Control & \LEFTcircle & \Circle & \CIRCLE & \Circle & \LEFTcircle & \LEFTcircle \\

RipTide \cite{RipTide-MICRO2022} & C/C++ & Steering Control & \LEFTcircle & \Circle & \CIRCLE & \Circle & \LEFTcircle & \LEFTcircle \\

TRIPS \cite{TRIPS-2004TACO} & C & Limited Predication & \Circle & \Circle & \Circle & \Circle & \Circle & \LEFTcircle \\

SARA \cite{SARA-ISCA2021} & Spatial & Limited Predication & \Circle & \Circle & \LEFTcircle & \Circle & \CIRCLE & \LEFTcircle \\

ICED \cite{ICED-MICRO2024} & C/C++ & Limited Predication & \LEFTcircle & \CIRCLE & \LEFTcircle & \Circle & \Circle & \LEFTcircle \\

Plaid \cite{Plaid-ASPLOS2025} & C/C++ & Limited Predication & \LEFTcircle & \LEFTcircle & \LEFTcircle & \Circle & \Circle & \LEFTcircle \\

\textbf{NEURA} & C/C++/MLIR & Hierarchical Predication & \CIRCLE & \CIRCLE & \CIRCLE & \CIRCLE & \CIRCLE & \CIRCLE \\

\bottomrule
\multicolumn{9}{l}{$\mathrm{\ast}$ Compiler extensibility for diverse microarchitectural features (e.g., specialized FUs).}\\
\multicolumn{9}{l}{$\mathrm{\dagger}$ Whether the framework can represent an entire kernel as a single, unified dataflow graph.}\\
\multicolumn{9}{l}{$\mathrm{\ddagger}$ Whether the framework can target both spatial-only and spatio-temporal execution models.}
\end{tabular}
}
\vspace*{-0.6\baselineskip}
\end{table}

The analysis in Sec. \ref{sec2:background} reveals that existing approaches fail to provide a unified, general, and retargetable approach for managing complex control flows across diverse CGRA architectures. In this section, we provide an overview of NEURA. What differentiates NEURA from existing works is its NEURA Dataflow IR. Its ability to flatten a kernel's CDFG into a DFG directly obviates the \textit{\textbf{Pitfall I}} (see Sec. \ref{sec5:transform}) issue of serialization via external controllers. NEURA Dataflow IR also decouples computational semantics from the CGRA execution model. This versatility allows NEURA to target both energy-efficient \textit{\textbf{spatial-only}} and high-performance \textit{\textbf{spatio-temporal}} execution models, overcoming the \textit{\textbf{Pitfall II}} (see Sec. \ref{sec7:implementation}) problems of rigid \textbf{\textit{spatial-only}} execution and tight hardware coupling. Unlike traditional predication, which lacks the semantics for nested control flow, NEURA introduces a novel predicated type system and predicate-management operations. This provides native support for hierarchical control dependencies, directly resolving the core representational flaw of \textit{\textbf{Pitfall III}} (see Sec. \ref{sec4:ir}) concerning the failure to unify hierarchical predicates.

\begin{figure}[h]
    \centering
    \vspace*{-0.15\baselineskip}
    \includegraphics[width=1.0\textwidth]{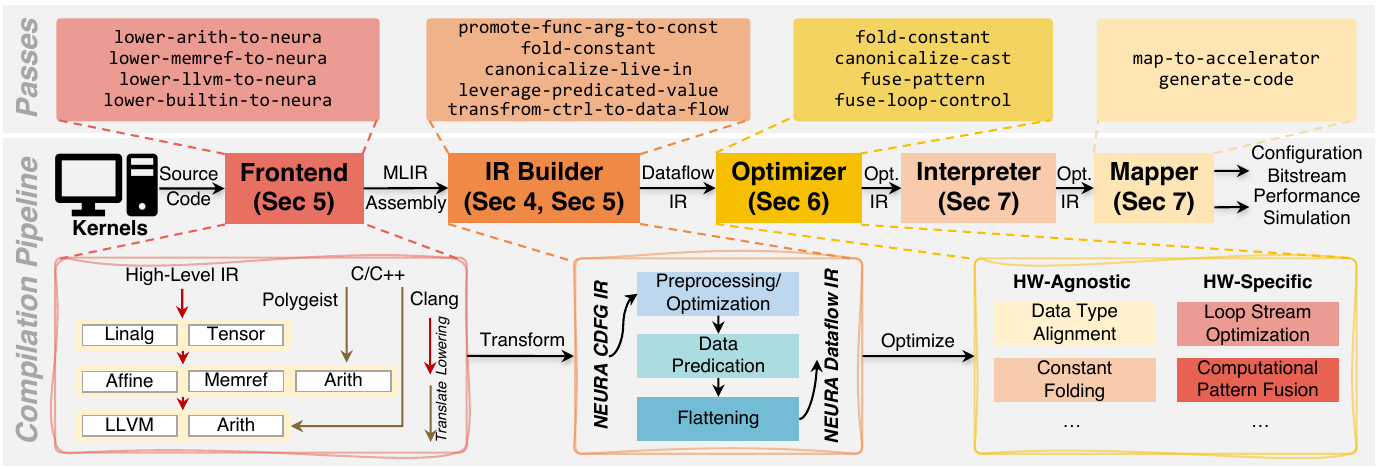}
    \vspace*{-0.5\baselineskip}
    \caption{Overview of the NEURA Compilation Flow -- The frontend accepts C/C++ and high-level IR kernels and lowers them into the NEURA CDFG IR. The IR builder then converts the kernel into the NEURA Dataflow IR through preprocessing, data predication, and flattening. The optimizer refines the dataflow IR using HW-Agnostic (e.g., constant folding) and HW-Specific optimizations (e.g., loop streaming). The optimized IR can be validated by the interpreter and processed by the mapper to get mapping results for configuration bitstream generation or performance simulation.
    }
    \label{fig:overview}
    \vspace*{-0.5\baselineskip}
\end{figure}

The NEURA compilation flow, as shown in Fig. \ref{fig:overview}, is built on the MLIR \cite{MLIR-CGO2021} framework. The flow begins at the extensible frontend. It supports any kernel input that can be lowered into standard dialects (e.g., \textit{llvm} \cite{LLVMDialect-2025Online} and \textit{arith} \cite{arithDialect-2025Online} dialects). For now, our frontend accepts C/C++ (via Clang \cite{Clang-2025Online} or Polygeist \cite{Polygeist-2021PACT}) and high-level IR like \textit{linalg} \cite{LinalgDialect-2025Online} and \textit{tensor} \cite{TensorDialect-2025Online} dialects and lowers them into our NEURA CDFG IR (see Sec. \ref{sec4:ir}). This IR preserves the conventional CFG structure for preprocessing. Next, the IR builder performs the core transformation (see Sec. \ref{sec5:transform}), systematically converting the NEURA CDFG IR into our novel NEURA Dataflow IR (see Sec. \ref{sec4:ir}), representing the kernel as a pure DFG. This dataflow IR then enters a two-phase optimizer, undergoing hardware-agnostic optimizations to safely optimize the DFG. This is followed by hardware-specific optimizations that leverage hardware-specific microarchitectural features to tailor the DFG for a specific CGRA (see Sec \ref{sec6:optimization}). The optimized IR can then be sent to the interpreter for functional validation. It can also be passed to the mapper, which takes hardware specifications and the IR to generate the mapping result for target CGRAs. The mapping result can drive a cycle-accurate simulator for fast evaluation or produce the configuration bitstream for target CGRAs.

\begin{figure}[t]
    \centering
    \includegraphics[width=1.0\textwidth]{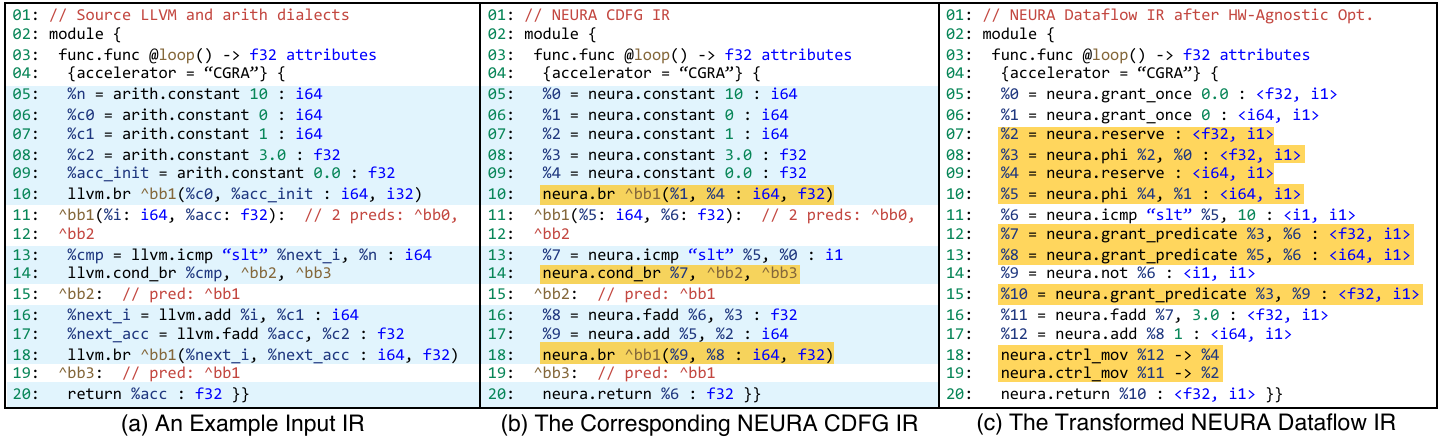}
    \vspace*{-1.4\baselineskip}
    \caption{A Kernel Example in NEURA --- (a) The input simple accumulation kernel represented in \textit{llvm} and \textit{arith} dialects. (b) The corresponding NEURA CDFG IR for further transformations. (c) The NEURA Dataflow IR leverages our predicated type system to represent the kernel in a pure dataflow manner. The \textcolor{myblue}{blue} blocks show the corresponding code between the input IR and the NEURA CDFG IR, while the \textcolor{myyellow}{yellow} blocks show the corresponding control-flow logic in CDFG IR and its dataflow representation.
    }
    \label{fig:example}
    \vspace*{-0.7\baselineskip}
\end{figure}

Fig. \ref{fig:example} shows how a simple kernel is represented during the NEURA compilation flow. The input starts as standard \textit{llvm} and \textit{arith} dialects, as shown in Fig. \ref{fig:example}(a). It is first converted to the NEURA CDFG IR, as shown in Fig. \ref{fig:example}(b), which preserves the explicit control flow via branch instructions (i.e., \texttt{br}, \texttt{cond\_br}). After lowering and optimization, the kernel becomes the pure NEURA Dataflow IR, as shown in Fig. \ref{fig:example}(c). This IR is built on our novel predicate type system, where every value's type is augmented with a predicate, denoted by \texttt{i1} (e.g., \texttt{<f32, i1>}). This predicate makes the data's validity an intrinsic property of the data. This mechanism, combined with specialized operations (e.g., \texttt{grant\_predicate}), allows the kernel's control flow to be flattened into data dependencies, enabling the kernel to be executed in a pure dataflow manner on a CGRA.

%% file: Ch4-IR.tex
\section{The NEURA Intermediate Representation}\label{sec4:ir}
The NEURA IR is designed to systematically bridge the \textbf{\textit{control-dataflow semantic gap}} between CDFG kernel representations and the native dataflow nature of CGRAs. It leverages the NEURA CDFG IR as a high-level intermediary to canonicalize standard compiler dialects (e.g., \textit{llvm} and \textit{arith} dialects). This representation is then systematically lowered into our NEURA Dataflow IR. This novel representation expresses the kernel, including its control flow, in a unified, pure dataflow manner. This section describes the role of the NEURA CDFG IR in Sec. \ref{sec4.1:cdfg-ir}, details the design of the NEURA Dataflow IR in Sec. \ref{sec4.2:dataflow-ir}, including its type system, operation set, and concludes in Sec. \ref{sec4.3:interface} by discussing its interface with the hardware Instruction Set Architecture (ISA).

\subsection{The NEURA CDFG IR: A Canonicalization Bridge for Dataflow Transformation}\label{sec4.1:cdfg-ir}
The transformation from a CDFG kernel representation to a pure dataflow representation requires converting control dependencies into data dependencies. Our approach is to embed control context into the data values using the predicated type system (see Sec. \ref{sec4.2:dataflow-ir}). However, a key challenge arises because data dependencies between BBs in a CDFG can be implicit (see Sec. \ref{sec5.2:preprocessing}). A value defined in one BB may be used several BBs later (e.g., \texttt{\%6} used in bb3 in Fig. \ref{fig:end2end-example}(b)), meaning it is live through intermediate BBs, but no data value is explicitly passed through them. Since our predicates must be attached to a data value, this absence of an explicit dataflow makes it hard to systematically propagate predicate information from one BB to the next. To solve this, we introduce the NEURA CDFG IR as a critical canonicalization stage. Its primary purpose is to provide a structured, LLVM-like representation that enables preprocessing passes to make the subsequent flattening simple and systematic. It adopts the conventional CDFG format and serves three key functionalities.

First, it provides a well-defined and extensible entry point for the core control-flow-to-dataflow transformation (see Sec. \ref{sec5:transform}), allowing NEURA to easily support new frontends. Second, it enables crucial preprocessing steps, specifically live-in canonicalization (see Sec. \ref{sec5.2:preprocessing}). This step transforms the NEURA CDFG IR to make inter-BB data dependencies explicit along CFG edges, as shown in Fig. \ref{fig:end2end-example}(d). This is essential because once every data dependency is explicit, predicate information can be systematically embedded (see Sec. \ref{sec5.4:flatenning}). Finally, our custom CDFG IR provides a natural home for domain-specific information. This allows us to enrich the IR by attaching CGRA-specific attributes to operations, such as metadata for constant handling (e.g., constant attributes in Fig. \ref{fig:end2end-example}(c)). Since these attributes lack semantic equivalents in standard dialects, this customizability makes the framework extensible to new CGRA features and optimizations.

\begin{table}[t]
\centering
\caption{A Representative List of Operations in the NEURA Dataflow IR}
\vspace*{-0.7\baselineskip}
\label{table:operations}
\resizebox{\linewidth}{!}{ 
\begin{tabular}{lp{6cm}lp{6cm}}
\toprule

\multicolumn{2}{c}{\textbf{Predicate Management Operations}} & \multicolumn{2}{c}{\textbf{Predicated States-Access Operations}} \\
\toprule
\texttt{neura.grant\_once(val)} & Assigns a single-use true predicate to an initial value \texttt{val}. & \texttt{neura.return(val)} & Terminates the program with \texttt{val} if its predicate is true. \\

\texttt{\multirow{2}{*}{\makecell[l]{neura.grant\_predicate\\(val, cond)}}} & Grants a new predicate to \texttt{val} based on the boolean data \texttt{cond}. & \texttt{neura.load(addr)} & Loads from \texttt{addr} only if the predicate is true. \\

\texttt{neura.phi(a, b, ...)} & Merges multiple dataflow paths by selecting the unique input with a true predicate. & \texttt{neura.store(val, addr)} & Stores \texttt{val} to memory only if both operands' predicates are true. \\

\texttt{neura.loop\_control$^\mathrm{\ast}$} & Fused operation generating the next index and validity predicate for a loop. & \texttt{\multirow{2}{*}{\makecell[l]{neura.load\_indexed$^\mathrm{\ast}$\\(base, <indices>)}}} & Fused operation combing address calculation and a predicated load. \\
\cmidrule[\heavyrulewidth](){1-4}
\multicolumn{2}{c}{\textbf{Predicated Computational Operations}} & \multicolumn{2}{c}{\textbf{Non-Materialized Structural Operations}}\\
\cmidrule[\heavyrulewidth](){1-4}

\texttt{neura.add(a, b)} & Computes the predicated sum of predicated value \texttt{a} and \texttt{b}. & \texttt{neura.reserve} &  Creates a placeholder for a value in a backward dataflow path.\\
\texttt{neura.icmp(cmp, a, b)} & Compares two predicated values \texttt{a} and \texttt{b}. Outputs a predicated boolean result. & \texttt{\multirow{2}{*}{\makecell[l]{neura.ctrl\_mov\\(val, placeholder)}}} & Defines a data dependency edge from \texttt{val} to \texttt{placeholder}. \\
\texttt{neura.muladd(a, b, c)$^\mathrm{\ast}$} & A hardware-specific multiply-add operation operating on predicated values. &  \\
\bottomrule
\multicolumn{4}{l}{$\mathrm{\ast}$ Hardware-specific operations, detailed in Sec. \ref{sec6.2:hw-specific}.}
\end{tabular}
}
\vspace*{-0.9\baselineskip}
\end{table}

\subsection{The NEURA Dataflow IR: A Predicated IR Enabling Hierarchical Predicates}\label{sec4.2:dataflow-ir}
The NEURA Dataflow IR is designed to holistically represent a kernel as a single, pure DFG. By transforming control flow into data dependencies, this representation enables kernels with complex control flows to be mapped and executed on CGRA fabrics. This is achieved through two foundational components --- a predicated type system and an operation set.

\subsubsection{The Predicated Type System: Making Control an Intrinsic Property of Data} The fundamental principle of the NEURA Dataflow IR is its predicated type system. Conventional CDFGs ensure correct kernel execution by explicitly dictating the sequence of operations and guarding state updates. Our predicated type system embeds control context directly into the data by making every data value carry its own validity information. This is achieved by uniformly pairing every data payload with a single predicate bit. As formalized in Fig. \ref{fig:type-system}, any value in the NEURA Dataflow IR is a tuple consisting of its data payload (e.g., an \texttt{f32}) and its predicate (i.e., a boolean \texttt{i1}). This predicate bit signifies whether the data payload is valid on the current execution path. If true, the data can be consumed by downstream operations to affect program state. If false, the data is considered invalid or nullified, and any operation consuming it is guaranteed not to produce side effects. This mechanism naturally handles complex control flows. For example, the predicate for an inner operation implicitly combines the validity derived from all enclosing control structures. This type system is the key to flattening the kernel into a pure dataflow representation.

\begin{figure}[h]
    \centering   
    \begin{tabular}{c}
    $\displaystyle\frac{\Gamma \vdash d : \tau \quad \Gamma \vdash p : \mathbb{B}}{\Gamma \vdash (d, p) : \tau_p}$ \quad \textsc{PredType}
    \end{tabular}
    \vspace*{-0.3\baselineskip}
    \caption{Predicated Type --- Let $\Gamma$ be a typing context. A value is predicated type $\tau_p$ in NEURA if it is a tuple $(d,p)$ pairing a data payload $d$ of standard type $\tau$ (e.g., \texttt{f32}, \texttt{i64}) with a predicate $p$ of the Boolean type $\mathbb{B}$.}
    \label{fig:type-system} 
    \vspace*{-0.8\baselineskip}
\end{figure}

\begin{figure}[t]
    \centering   
    \resizebox{\linewidth}{!}{ 
    \begin{tabular}{c@{\quad}c}
    $\displaystyle\frac{\langle v_{addr}, \sigma\rangle \Downarrow \langle(d_{addr}, \text{true}), \sigma\rangle}{\langle \textbf{load}(v_{addr}), \sigma\rangle \Downarrow \langle(\sigma(d_{addr}), \text{true}), \sigma\rangle}$ \quad \textsc{Load} & $\displaystyle\frac{\langle v_{data}, \sigma\rangle \Downarrow \langle(d_{data}, \text{true}), \sigma\rangle \quad \langle v_{addr}, \sigma\rangle \Downarrow \langle(d_{addr}, \text{true}), \sigma\rangle}{\langle \textbf{store}(v_{data}, v_{addr}), \sigma\rangle \Downarrow \langle(), \sigma[d_{addr}:=d_{data}]\rangle}$ \quad \textsc{Store} \\
    \addlinespace\addlinespace\addlinespace
    \multicolumn{2}{c}{$\displaystyle\frac{\langle v_1, \sigma \rangle \Downarrow \langle(d_1, p_1), \sigma\rangle \quad \langle v_2, \sigma \rangle \Downarrow \langle(d_2, p_2), \sigma\rangle \quad d_{res}=d_1\oplus d_2 \quad p_{res}=p_1\land p_2}{\langle v_1 \bm{\oplus_p} v_2, \sigma\rangle \Downarrow \langle (d_{res}, p_{res}), \sigma \rangle}$ \quad \textsc{PredComp}}\\
    \addlinespace\addlinespace\addlinespace
    \multicolumn{2}{c}{$\displaystyle\frac{\langle v_{in}, \sigma\rangle \Downarrow \langle(d_{in}, p_{in}), \sigma\rangle \quad \sigma(\text{state}_{id})=\text{fresh}}{\langle \textbf{grant\_once}_{id}(v_{in}), \sigma\rangle \Downarrow \langle(d_{in}, \text{true}), \sigma[\text{state}_{id}:=\text{consumed}]\rangle}$ \quad \textsc{GrantOnce}} \\
    \addlinespace\addlinespace\addlinespace
    \multicolumn{2}{c}{ $\displaystyle\frac{\langle v_{val}, \sigma\rangle \Downarrow \langle(d_{val}, p_{val}), \sigma\rangle \quad \langle v_{cond}, \sigma\rangle \Downarrow \langle(d_{cond}, p_{cond}), \sigma\rangle}{\langle \textbf{grant\_predicate}(v_{val}, v_{cond}), \sigma\rangle \Downarrow \langle(d_{val}, d_{cond}\land p_{cond}), \sigma\rangle}$ \quad \textsc{GrantPred}} \\
    \addlinespace\addlinespace\addlinespace
    \multicolumn{2}{c}{$\displaystyle\frac{\forall i\in\{1\dots n\}: \langle v_i, \sigma \rangle \Downarrow \langle(d_i, p_i), \sigma\rangle \quad \exists! k\in \{1\dots n\}:p_k=\text{true}}{\langle \textbf{phi}(v_1, \dots,v_n), \sigma \rangle \Downarrow \langle(d_k, \text{true}), \sigma\rangle}$ \quad \textsc{Phi}}
    \end{tabular}
    }
    \vspace*{-0.2\baselineskip}
    \caption{A Selection of Big-Step Structural Operational Semantics for NEURA Dataflow IR --- The relation $\Downarrow$ defines the evaluation of expressions. The state $\sigma$ is an environment that includes memory and the internal state of \texttt{grant\_once} operations. A value $v$ is a tuple $(d,p)$. In \textsc{PredComp} rule, $\oplus_p$ denotes the predicated computational operations while $\oplus$ denotes the standard computational operations on the data payloads.}
    \label{fig:semantics} 
    \vspace*{-1.1\baselineskip}
\end{figure}

\subsubsection{The NEURA Operation Set: From Traditional to Hierarchical Predication}\label{sec4.2.2:operation-set} 
The NEURA Dataflow IR is composed of a set of operations designed to compute on predicated values, manage state, and explicitly handle the predicate logic that replaces traditional control flow. As detailed in Table \ref{table:operations}, these operations fall into four categories: (1) \textbf{Predicated Computational Operations} include standard arithmetic and logical operations lifted to operate on predicated values (\textsc{PredType} in Fig. \ref{fig:type-system}); (2) \textbf{Predicated State-Access Operations} interact with stateful resources, such as memory and the kernel's termination state, where physical side effects are guaranteed to occur only if all the operands are valid; (3) \textbf{Predicate Management Operations} contain the specialized operations that create, manipulate, and merge predicates to encode control flow; (4) \textbf{Non-Materialized Structural Operations} are used to structure the dataflow graph (e.g., loop recurrences) and maintain Static Single Assignment (SSA) properties \cite{SSA-PPoPP1989} without corresponding to physical hardware (see Sec. \ref{sec5.4:flatenning}).

Fig. \ref{fig:semantics} presents a selection of big-step structural operational semantics for the NEURA Dataflow IR, defining the evaluation judgment and clarifying its behaviors. The rules illustrate how the predicated type system is operationalized. For example, the \textsc{PredComp} rule shows that computational validity is contingent on the operand validity. The \textsc{Load} and \textsc{Store} rules stipulate that an operation with a false predicate operand does not satisfy the premises, thus producing no side effect on kernel state. The \textsc{GrantOnce}, \textsc{GrantPred}, and \textsc{Phi} rules formally define the mechanisms for creating, transforming, and merging predicated values for control flow encoding.

\subsection{Interface with the Hardware ISA and Retargetability}\label{sec4.3:interface}

A principle of the NEURA Dataflow IR is its hardware-conscious design. The IR is designed to be highly expressive while requiring minimal, low-cost extensions to CGRAs' ISA. This design philosophy ensures NEURA can be readily adopted by existing or future CGRAs.

\begin{wrapfigure}{r}{0.45\textwidth}
  \vspace*{-0.8\baselineskip}
  \centering
  \includegraphics[width=0.45\textwidth]{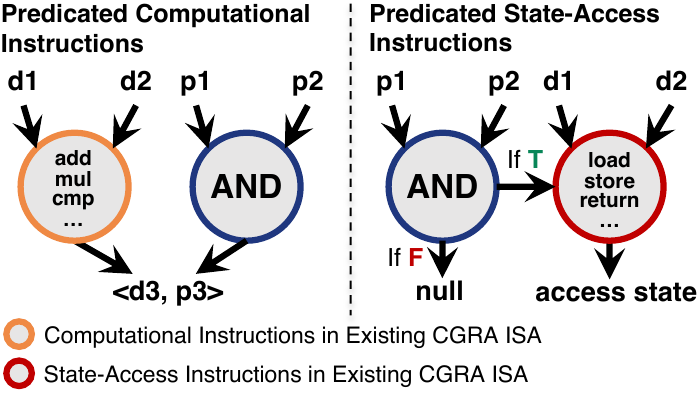}
  \vspace*{-1.0\baselineskip}
  \caption{Extending the Computational and State-Access Instructions in Existing ISA.
  }
  \label{fig:isa-extension}
  \vspace*{-0.5\baselineskip}
\end{wrapfigure}
\textbf{Modifications to Existing ISA.} NEURA requires only minor and uniform modifications to the existing CGRA ISA. As shown in Fig. \ref{fig:isa-extension}, \textbf{computational instructions} are augmented to accept predicate bits. They perform the original computation in parallel with a logical AND of all incoming predicates to produce the output predicate. \textbf{State-access instructions} (e.g., \texttt{load}, \texttt{store}, \texttt{return}) use the logical AND of the input predicates as an enable signal, triggering side effects (e.g., memory read/write) only if it is true. These modifications are low-cost, typically requiring one additional AND gate per FU.

\textbf{Introduction of New ISA Instructions.} To manage the predicated type system, NEURA introduces only three new specialized instructions, \texttt{grant\_once}, \texttt{grant\_predicate}, and \texttt{phi} (see Sec. \ref{sec4.2.2:operation-set}). Predicated computational and state-access instructions propagate existing predicate bits (typically via an AND gate). In contrast, these new instructions generate and manipulate the predicate bits. These three instructions are the essential hardware primitives that enable the NEURA compiler to systematically flatten complex control flow into pure dataflow (see Sec. \ref{sec5:transform}).

\textbf{Retargetability.} The required ISA extensions are minimal, comprising three new instructions and a systematic modification to existing computational and state-access instructions. This design is highly adaptable to diverse CGRAs, including those with specialized microarchitectural features (e.g., specialized FUs). The NEURA Dataflow IR is a pure DFG that only defines logical computations and their dependencies, decoupling the representation from the execution model. This also naturally enables the IR to be executed by both \textbf{\textit{spatio-temporal}} and \textbf{\textit{spatial-only}} execution models. Furthermore, facilitated by NEURA's type system and dataflow semantics, we also implemented a compiler pass to transform the NEURA Dataflow IR into a steering control representation, ensuring NEURA's compatibility with existing steering-based architectures (e.g., RipTide \cite{RipTide-MICRO2022}). Together, these properties provide NEURA with comprehensive \textbf{retargetability}.

%% file: Ch5-Transform.tex
\section{Lowering to NEURA Dataflow IR}\label{sec5:transform}

\begin{figure}[t]
    \centering
    \includegraphics[width=1.0\textwidth]{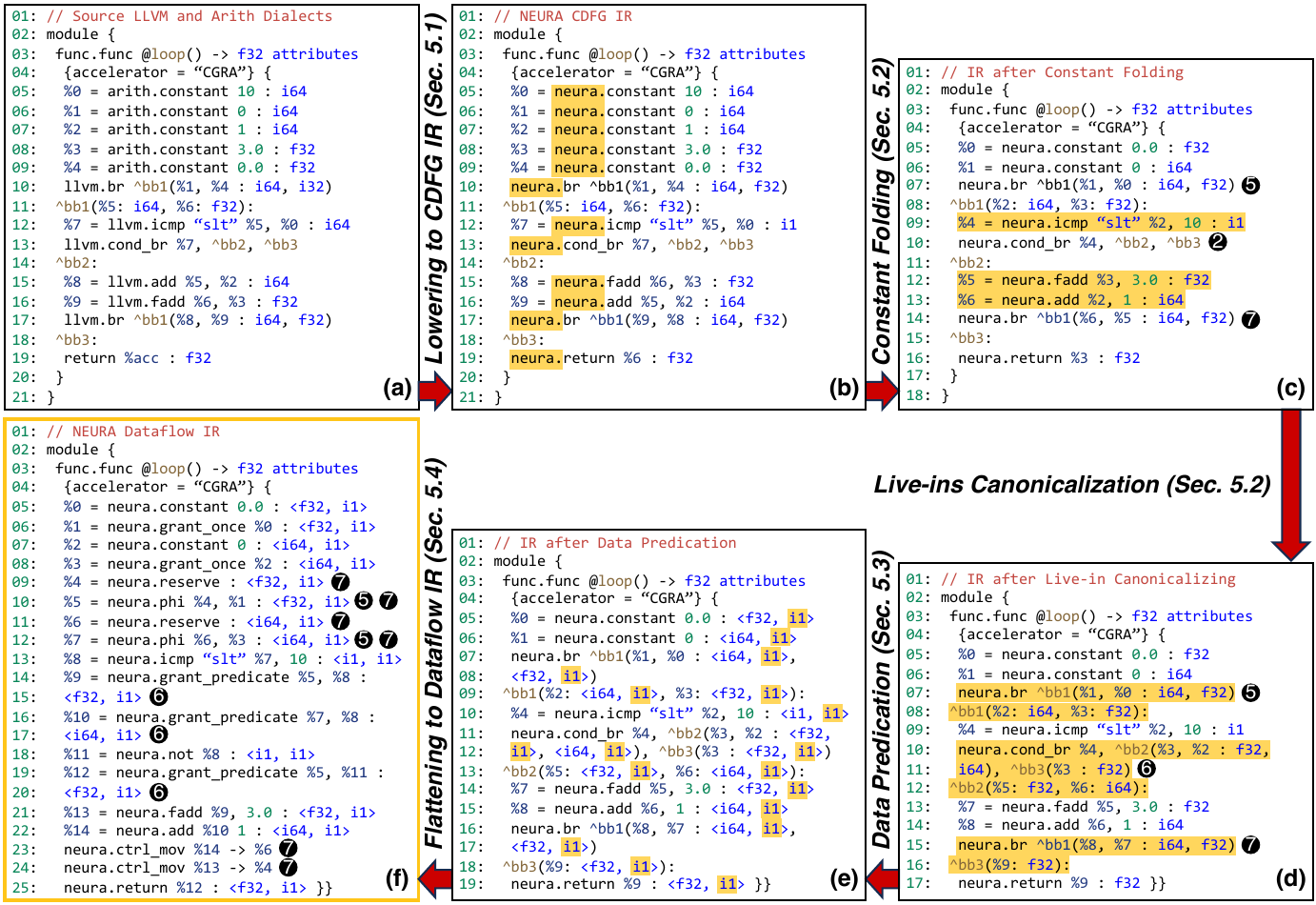}
    \caption{An End-to-End Example of the NEURA Lowering Process --- This figure illustrates the sequential transformation for a simple kernel. (a) standard MLIR dialects are lowered to (b) NEURA CDFG IR (Sec. \ref{sec5.1:initial-lowering}). The IR is then preprocessed via (c) constant folding and (d) live-in canonicalization (Sec. \ref{sec5.2:preprocessing}). This is followed by (e) data predication (Sec. \ref{sec5.3:predication}) and finally flattened into (f) the NEURA Dataflow IR (Sec. \ref{sec5.4:flatenning}). The \textcolor{myyellow}{yellow} blocks indicate the code delta from the preceding phase. The circled numbers (e.g., \protect\circledblack{2}, \protect\circledblack{5}, \protect\circledblack{6}, \protect\circledblack{7}) annotate the CFG edge types in (c) and (d), and track their transformation into dataflow operations in (f). The CFG edge types are defined in Fig. \ref{fig:rewriting-rules}. 
    }
    \label{fig:end2end-example}
    \vspace*{-1.3\baselineskip}
\end{figure}
A key contribution of NEURA is its systematic methodology to transform kernels represented in a conventional CDFG manner into our pure dataflow representation. This transformation is the core of NEURA to bridge the \textbf{\textit{control-dataflow semantic gap}} between the CDFG kernel representation and the dataflow nature of CGRA fabrics. Built upon the MLIR infrastructure, this process is structured as a sequence of well-defined passes, as outlined in Fig. \ref{fig:overview}. This section details the four primary stages of this lowering process: initial translation to NEURA CDFG IR (Sec. \ref{sec5.1:initial-lowering}), preprocessing (Sec. \ref{sec5.2:preprocessing}), systematic data predication (Sec. \ref{sec5.3:predication}), and final flattening into the NEURA Dataflow IR (Sec. \ref{sec5.4:flatenning}). Fig. \ref{fig:end2end-example} provides an end-to-end example to illustrate this lowering process.

\subsection{Lowering to NEURA CDFG IR}\label{sec5.1:initial-lowering}
This stage transforms the kernel from source code or high-level IR into our NEURA CDFG IR. We first leverage existing frontends (e.g., Clang \cite{Clang-2025Online} and Polygeist \cite{Polygeist-2021PACT}) to lower kernels into standard MLIR dialects (see Fig. \ref{fig:end2end-example}(a)), like \textit{llvm} \cite{LLVMDialect-2025Online}, \textit{arith} \cite{arithDialect-2025Online}, and \textit{memref} \cite{memrefDialect-2025Online}. Then, as illustrated in Fig. \ref{fig:overview}, a set of custom conversion passes transforms these standard dialects into their semantic equivalents in the NEURA CDFG IR (e.g., \texttt{arith.addi} and \texttt{llvm.br} become \texttt{neura.add} and \texttt{neura.br}, respectively). This conversion establishes a common, CDFG-based starting point tailored for subsequent transformations, unifying diverse inputs while preserving the original CDFG structure. The resulting NEURA CDFG IR (see Fig. \ref{fig:end2end-example}(b)) serves as the input for the preprocessing stage.

\begin{figure}[t]
    \centering
    \includegraphics[width=1.0\textwidth]{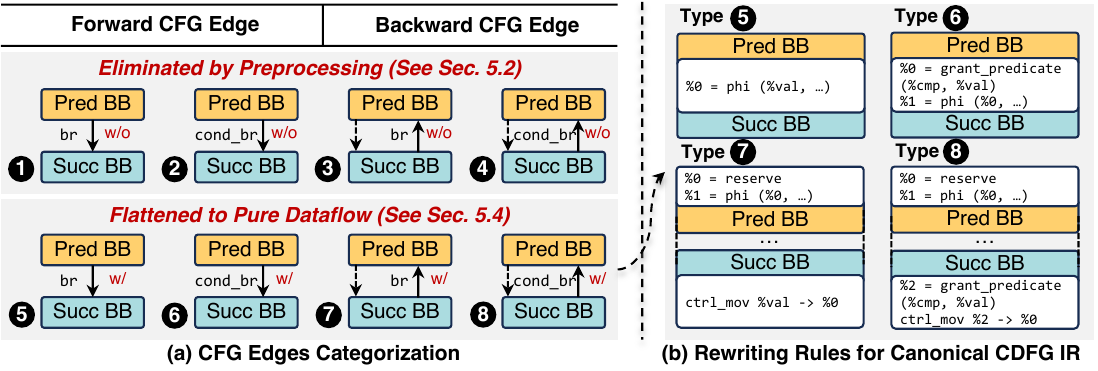}
    \vspace*{-0.8\baselineskip}
    \caption{CFG Edge Categorization and Rewriting Rules --- (a) CFG edges in a kernel can be categorized into eight types. Edges without explicit value passing (type \protect\circledblack{1} - \protect\circledblack{4}) are eliminated by preprocessing. (b) We can convert the remaining four types of edges into pure dataflow dependencies via deterministic rewriting rules.}
    \label{fig:rewriting-rules}
    \vspace*{-0.7\baselineskip}
\end{figure}
\subsection{Preprocessing on NEURA CDFG IR}\label{sec5.2:preprocessing}
Before flattening to a pure dataflow representation, the NEURA CDFG IR undergoes several preprocessing passes. These passes are designed to simplify the IR's structure and establish a canonical form for the subsequent flattening stage. While this stage also includes hardware-agnostic optimizations (e.g., constant folding, detailed in Sec. \ref{sec6:optimization}), this section focuses specifically on two canonicalization processes that are critical for enabling the dataflow flattening.

First, in NEURA, kernels targeted for CGRA acceleration are represented as functions. To handle their parameters consistently with the dataflow model, we apply the \textbf{function arguments promotion}. This transformation replaces all uses of a function's arguments with \texttt{neura.constant} operations materialized in the entry BB. This ensures all inputs to the function are treated uniformly as SSA values defined by operations, simplifying subsequent analysis and transformations.

Second, we perform \textbf{live-ins canonicalization} to make inter-BB data dependencies explicit. This explicit representation is a prerequisite for systematically embedding predicate information onto data values during the flattening stage. As shown in Fig. \ref{fig:rewriting-rules}(a), we classify CFG edges into eight types based on their direction (forward or backward), terminator type (\texttt{br} or \texttt{cond\_br}), and whether they explicitly carry SSA values (w/ or w/o values). These different edge types create a non-uniform structure that complicates further flattening. Some edges do not explicitly pass SSA values (e.g., type \circledblack{2} edge in Fig. \ref{fig:end2end-example}(c)), which hinders predicate propagation. Others may carry values, but fail to propagate the complete set of live-ins required by the successor, leading to incorrect predication. To establish a uniform and complete mechanism for CDFG flattening, our \texttt{canonicalize-live-in} pass, detailed in Algorithm \ref{alg:canonicalize-live-in}, transforms the CDFG into a canonical form where every control flow edge explicitly passes the complete set of live-in values to its target block. For example, this pass transforms the type \circledblack{2} edge in Fig. \ref{fig:end2end-example}(c) into a type \circledblack{6} edge in Fig. \ref{fig:end2end-example}(d).

The \texttt{canonicalize-live-in} pass operates in two phases, as shown in Algorithm \ref{alg:canonicalize-live-in}. \textbf{Phase 1} performs a fixed-point analysis to compute the complete transitive live-ins for each BB. It initializes the live-in sets based on direct uses (lines 2-3) and then iteratively propagates liveness information through the CFG until convergence (lines 4-11). Specifically, for each BB, it examines its successors' live-ins (line 7). If a successor needs a value not defined in the current BB, that value is added to the current BB's live-in set (lines 9-10). This process repeats until a fixed point is reached, guaranteeing the analysis captures the complete transitive live-ins for each BB. \textbf{Phase 2} transforms the CDFG based on the analysis results. For each BB with a non-empty live-in set (line 13), the pass first adds new block arguments corresponding to these live-ins (line 15), replaces uses of the original live-in value with these new arguments (lines 6), and rewrites the terminators of all predecessor blocks to explicitly pass the required values when branching to this BB (lines 18-19). This transformation is proven complete in Theorem \ref{the:live-in}.

\algtext*{EndFor}
\algtext*{EndIf}
\algtext*{EndWhile}
\algrenewcommand\algorithmicrequire{\textbf{Input:}}
\algrenewcommand\algorithmicensure{\textbf{Output:}}
\begin{algorithm}[t]
\footnotesize
\begin{spacing}{0.9}
\caption{The Canonicalize Live-in Pass}
\label{alg:canonicalize-live-in}
\begin{algorithmic}[1]
    \Require A function $F$ represented by a CDFG.
    \Ensure A canonical CDFG representation with all inter-block data dependencies explicitly passed along CFG edges.
    \State Let $L_{in}$ be a map from a block to the set of live-in values; \hfill \Comment{\textbf{Phase 1}: Fixed-point live-in analysis}
    \ForAll{block $B \in F$ where $B$ is not the entry block} \Comment{Initialize the live-in set for each basic block (BB)}
        \State $L_{in}[B] \leftarrow$ ComputeDirectLiveIns($B$);
    \EndFor
    \While{$changed$} \Comment{Compute the complete live-ins for each BB}
        \State $changed \leftarrow false$;
        \ForAll{block $B \in F$ where $B$ is not the entry block}
            \ForAll{successor $S$ of $B$}
                \State $L_{propagate} \leftarrow L_{in}[S] \setminus \text{DefinedIn}(B)$;
                \If{$L_{propagate} \not\subseteq L_{in}[B]$}
                    \State $L_{in}[B] \leftarrow L_{in}[B] \cup L_{propagate}$;
                    \State $changed \leftarrow true$;
                \EndIf
            \EndFor
        \EndFor
    \EndWhile
    \State Let $A_{new}$ be an empty list of new arguments for $B$; \hfill \Comment{\textbf{Phase 2}: CDFG transformation}
    \ForAll{block $B \in F$ where $L_{in}[B]$ is not empty}
        \ForAll{value $v \in L_{in}[B]$} \Comment{Replace live-ins with block arguments}
            \State $arg \leftarrow$ AddBlockArgument($B$, TypeOf($v$));
            \State ReplaceAllUseOf($v$, $arg$, within $B$);
            \State $A_{new}$.append(arg);
        \EndFor
        \ForAll{predecessor $P$ of $B$} \Comment{Modify the terminator of predecessors}
            \State RewriteTerminator(GetTerminator($P$), $B$, $A_{new}$);
        \EndFor
    \EndFor
\end{algorithmic}
\end{spacing}
\end{algorithm}

\begin{theorem}
\label{the:live-in}
Given an input function $F$ represented in NEURA CDFG IR containing no semantically irrelevant (i.e., dead) code, the \texttt{canonicalize-live-in} pass transforms $F$ into an equivalent representation $F'$. This transformation guarantees that for every control flow edge $e=(B_{pred}, B_{succ})$ from a predecessor basic block $B_{pred}$ to a successor basic block $B_{succ}$ in $F'$, the terminator of $B_{pred}$ explicitly passes the complete live-in set $LiveIn(B_{succ})$ as block arguments to $B_{succ}$.
\end{theorem}

\begin{proof}
Consider an arbitrary control-flow edge $e=(B_{pred}, B_{succ})$ in the transformed CDFG representation $F'$. Let $LiveIn(B_{succ})$ be the live-in set of $B_{succ}$ computed by the fixed-point analysis in Phase 1 of Algorithm \ref{alg:canonicalize-live-in}. We analyze two mutually exclusive cases on this set:

\textbf{Case 1: $LiveIn(B_{succ})\neq \emptyset$} --- In this case, the CDFG transformation (Phase 2) explicitly adds block arguments to $B_{succ}$ for all $v\in LiveIn(B_{succ})$ and rewrites the terminator of $B_{pred}$ to pass these values along edge $e$. Thus, the edge $e$ satisfies the theorem's property by construction.

\textbf{Case 2: $LiveIn(B_{succ})= \emptyset$} --- An empty live-in set for $B_{succ}$ implies that no value defined outside of $B_{succ}$ is used within it or is needed by any of its successors. For any reachable successor $B_{succ}$ in a function without dead code, this condition cannot hold. As the theorem premise excludes programs without dead code, this case is precluded for any relevant edge.

For any edge in a semantically meaningful program, Case 1 must hold. Thus, the pass is guaranteed to transform all such edges to explicitly pass the complete live-in set of the target BB.
\end{proof}

\subsection{Data Predication}\label{sec5.3:predication}
This stage systematically lifts the entire kernel into our predicated type system (see Sec. \ref{sec4.2:dataflow-ir}). The \texttt{leverage-predicated-value} pass converts the type of every SSA value from a standard type $\tau$ to its predicated counterpart $\tau_p$ (e.g., \texttt{f32} becomes \texttt{<f32, i1>}) and replaces each standard operation with its predicated version that operates on these new predicated types. The result of this stage, as illustrated in Fig. \ref{fig:end2end-example}(e), is a kernel where every value explicitly carries a validity predicate.

\subsection{Flattening to NEURA Dataflow IR}\label{sec5.4:flatenning}
The final stage, accomplished by \texttt{transform-ctrl-to-data-flow} pass, flattens the NEURA CDFG IR into the NEURA Dataflow IR. Leveraging the canonical form from preprocessing (see Sec. \ref{sec5.2:preprocessing}), this pass only needs to handle the four CFG edge types that explicitly carry values (type \circledblack{5} - \circledblack{8} edges in Fig. \ref{fig:rewriting-rules}(a)). Deterministic rewriting rules are applied to each of these edges, transforming control dependencies into data dependencies using our predicate management operations and non-materialized structural operations, as illustrated in Fig. \ref{fig:rewriting-rules}(b). A key challenge arises when flattening backward edges (type \circledblack{7} - \circledblack{8} edges). Representing such backward dependencies directly in a DFG risks violating SSA form \cite{SSA-PPoPP1989}, as a value might appear used before it is defined. To address this while strictly maintaining SSA, we introduce non-materialized structural operations (see Sec. \ref{sec4.2:dataflow-ir}). First, \texttt{neura.reserve} defines a placeholder SSA value, acting as a forward declaration. Then, \texttt{neura.ctrl\_mov} establishes the backward data dependency edge to update this placeholder. By applying these rules, all branch operations are eliminated, flattening all BBs into a single block. The final result, as shown in Fig. \ref{fig:end2end-example}(f), is a kernel represented by the NEURA Dataflow IR.

%% file: Ch6-Optimization.tex
\section{Integrating Hardware Agnostic and Specific Optimizations}\label{sec6:optimization}

NEURA integrates optimizations at multiple stages of its compilation pipeline. A key strength of NEURA's MLIR-based design is its extensibility. This allows the NEURA frontend to accept high-level MLIR dialects (e.g., \textit{linalg} and \textit{affine} dialects) and naturally benefit from their upstream optimizations, such as loop tiling or loop fusion. In addition, NEURA integrates its own optimization passes, which are the focus of this section. We classify these NEURA internal passes into two categories: \textbf{Hardware-Agnostic Optimizations} (see Sec. \ref{sec6.1:hw-agnostic}) for logical simplification of the IR, and \textbf{Hardware-Specific Optimizations} (see Sec. \ref{sec6.2:hw-specific}) to tailor the IR for target CGRA microarchitectural features. This section presents representative passes for each category, but NEURA remains extensible, allowing users to easily define their own custom optimizations. More optimizations and implementation details are available in our open-source release.

\subsection{Hardware-Agnostic Optimization}\label{sec6.1:hw-agnostic}
These optimizations are hardware-agnostic, performing logically equivalent transformations on the NEURA CDFG/Dataflow IR. They aim to produce a more efficient IR by reducing redundancy and canonicalizing data types, and are portable across all supported CGRA targets.

\textbf{Data Type Alignment.} Frontends often introduce abstract data types that lack the specific bit width required by hardware. Our \texttt{canonicalize-cast} pass resolves this by canonicalizing these abstract data types (e.g., \texttt{index}) into concrete base types (e.g., \texttt{i32} or \texttt{i64}) based on user specification. This transformation ensures all data types are explicitly sized before hardware-specific stages, which is a prerequisite for correct mapping and resource allocation.

\textbf{Constant Folding.} In standard MLIR, constants are materialized via dedicated \texttt{constant} operations \cite{LLVMConstant-2025Online}. This approach is inefficient for CGRAs because each DFG operation maps to a physical hardware tile on CGRAs, and dedicating a compute tile solely to produce or forward a constant value significantly wastes hardware resources. These explicit constant operations also add unnecessary nodes and edges to the DFG, increasing mapping complexity. The \texttt{fold-constant} pass resolves this by embedding constant operands as attributes (i.e., immediate values) directly within their consuming operations, simplifying the DFG and freeing hardware resources. Fig. \ref{fig:example}(c) shows its application on the NEURA Dataflow IR, and Fig. \ref{fig:end2end-example}(c) shows its use during the preprocessing stage.

\subsection{Hardware-Specific Optimization}\label{sec6.2:hw-specific}
These passes specialize the general NEURA Dataflow IR by leveraging CGRA \textbf{microarchitectural features}, which showcases NEURA's retargetability. NEURA facilitates this specialization through MLIR's pattern-matching and rewriting infrastructure, allowing new hardware-specific operations to be easily integrated. We demonstrate this capability through two example optimizations.

\textbf{Computational Pattern Fusion.} Many CGRAs provide specialized FUs that fuse common operation patterns into a single, more efficient instruction (e.g., \texttt{mul-add}). Our \texttt{fuse-pattern} pass identifies patterns in our IR that match these predefined patterns and replaces them with corresponding hardware-specific fused operations. As shown in Table \ref{table:operations}, an address calculation followed by a memory access (i.e., \texttt{neura.load}) can be fused into a \texttt{neura.load\_indexed} operation. Similarly, a \texttt{neura.mul} followed by a \texttt{neura.add} can be replaced with a \texttt{neura.muladd} operation. This optimization reduces the number of DFG nodes, shortens critical paths, and enables more efficient mapping.

\begin{wrapfigure}{r}{0.44\textwidth}
  \vspace*{-0.3\baselineskip}
  \centering
  \includegraphics[width=0.44\textwidth]{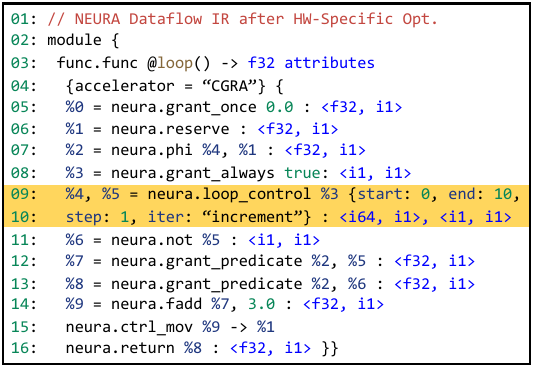}
  \caption{NEURA Dataflow IR of Fig. \ref{fig:end2end-example}(f) after applying Loop Streaming Optimization.}
  \label{fig:hw-specific}
  \vspace*{-0.5\baselineskip}
\end{wrapfigure}

\textbf{Loop Streaming Optimization.} As shown in Fig. \ref{fig:end2end-example}(f), while the general NEURA Dataflow IR represents loop control logic using \texttt{phi}, \texttt{add}, and \texttt{icmp} operations, this brings inter-iteration dependencies that may bottleneck the mapping II \cite{HyCUBE-DAC2017}. For CGRAs supporting loop stream operations, our \texttt{fuse-loop-control} pass recognizes static-bounded loop patterns and replaces the loop control logic with a \texttt{neura.loop\_control} operation, as shown in Fig. \ref{fig:hw-specific}. This new operation encapsulates the loop control logic (e.g., index update and boundary check) into a single loop streaming FU. This optimization breaks the inter-iteration bottleneck in the DFG, enabling more aggressive pipelining and achieving higher hardware utilization for common loop structures.

%% file: Ch7-Implementation.tex
\section{Implementation}\label{sec7:implementation}
The NEURA framework is implemented in 15K lines of C++ code on top of the MLIR infrastructure. We leverage MLIR's modularity and its powerful features for dialect definition, transformation, and optimization. This section details the implementation of NEURA's key components.

\textbf{Frontends.} To ensure broad applicability, NEURA compilation flow starts from kernels lowered into standard MLIR dialects (e.g., \textit{llvm} and \textit{arith} dialects). This common starting point allows it to readily support various frontends, including C/C++ (via Clang \cite{Clang-2025Online} and Polygeist \cite{Polygeist-2021PACT}) and high-level IRs like the \textit{linalg} \cite{LinalgDialect-2025Online} and \textit{tensor} \cite{TensorDialect-2025Online} dialects. The framework is designed for future extension, with plans to incorporate support for PyTorch \cite{PyTorch-2024ASPLOS} via Torch-MLIR \cite{Torch-MLIR-2025Online}.

\textbf{Core Compiler Infrastructure.} The core of NEURA is a comprehensive dialect defined in MLIR, defining operations for both the NEURA CDFG IR and Dataflow IR (Sec. \ref{sec4:ir}). The transformation logic (Sec. \ref{sec5:transform}) is implemented as a series of MLIR conversion passes, and hardware-agnostic and hardware-specific optimizations (Sec. \ref{sec6:optimization}) are built upon MLIR's pattern-rewriting mechanism. This use of pattern-rewriting enables a highly extensible optimization pipeline, allowing users to easily define and integrate new optimization patterns without modifying the core compiler infrastructure.

\textbf{IR Validation and Backend Support.} NEURA includes backend components to verify IR correctness and target CGRAs. For functional verification before hardware deployment or simulation, we have implemented an \textbf{interpreter} that interprets the NEURA Dataflow IR in a dataflow-driven manner, validating the correctness of the dataflow IR. Additionally, a heuristic-based mapper, adapted from OpenCGRA \cite{OpenCGRA-ICCD2020}, maps the NEURA Dataflow IR onto the CGRAs. This mapper accepts target architecture specifications, supporting both \textit{\textbf{spatial-only}} and \textit{\textbf{spatio-temporal}} execution models, and aims to find a valid mapping result with the minimum II. A flexible API is also provided to allow for the integration of diverse mapping algorithms.

%% file: Ch8-Evaluation.tex
\section{Experiments}\label{sec9:evaluation}
We present a comprehensive evaluation to demonstrate the effectiveness of NEURA. Our evaluation first quantifies the low area overhead of our ISA extensions (Sec. \ref{sec8.15:hardware-overhead}). It then demonstrates NEURA's SOTA performance on a high-performance \textbf{\textit{spatio-temporal}} architecture (Sec. \ref{sec8.2:control-eval}) and a competitive solution on low-power \textbf{\textit{spatial-only}} architecture (Sec. \ref{sec8.25:low-power}). We analyze the impact of hw-agnostic/-specific optimizations (Sec. \ref{sec8.3:optimization-expr}) and show the scalability of NEURA Dataflow IR (Sec. \ref{sec8.4:scalability}). NEURA's ability to support both \textbf{\textit{spatio-temporal}} and \textbf{\textit{spatial-only}} execution models, combined with its extensibility for diverse microarchitectural features, validates its \textbf{retargetability}.

\subsection{Experiment Settings}\label{sec8.1:expr-setting}
\textbf{Evaluation Architectures.} We develop two prototype CGRAs in RTL \cite{OpenCGRA-ICCD2020, VecPAC-2023ICCAD} designed to execute the general NEURA Dataflow IR (i.e., w/o specialized FU). These architectures, referred to as NEURA-SO and NEURA-ST, support the \textbf{\textit{spatial-only}} and \textbf{\textit{spatio-temporal}} execution models, respectively. Both are designed as typical CGRAs \cite{VecPAC-2023ICCAD, OpenCGRA-ICCD2020}, featuring a grid of tiles interconnected by a King Mesh NoC \cite{AURORA-2021DATE} and incorporating only the minimal ISA extensions required by Sec. \ref{sec4.3:interface}.

\textbf{Baselines.} We benchmark NEURA against three SOTA frameworks, each representing one of the primary control flow handling strategies discussed in Sec. \ref{subsec2.2:controlflow}: Marionette \cite{Marionette-MICRO2023} (The CDFG Strategy), RipTide \cite{RipTide-MICRO2022} (The Steering Control Strategy), and ICED \cite{ICED-MICRO2024} (The Limited Predication Strategy). For a fair comparison, ICED's dynamic power management features have been removed. To justify that the observed differences stem from IR contributions, rather than implementation-specific details (e.g., operating frequencies), we normalize all evaluated architectures to operate at 800MHz.

\textbf{Benchmarks.} We collect a diverse suite of kernel-level benchmarks from PolyBench \cite{PolyBench-2014IMPACT}, MachSuite \cite{MachSuite-2014IISWC}, CGRA-Bench \cite{OpenCGRA-ICCD2020}, and CHStone \cite{CHStone-IMT2009} that offer a wide spectrum of control flow features (e.g., nested loops, branches in loops). As summarized in Table \ref{table:benchmarks}, these benchmarks are categorized into four application domains to systematically evaluate performance on different computational patterns. In addition to these isolated kernels, we evaluate application-level performance using two real-world applications composed of multiple interacting kernels: a 2-layer Graph Convolutional Network (GCN) derived from PyTorch-Geometric \cite{PyTorch-Geomatric-axiv2019}, and Lower-Upper (LU) Decomposition sourced from CGRA-Bench \cite{OpenCGRA-ICCD2020}.

\begin{wraptable}{r}{0.679\textwidth}
\vspace*{-1.0\baselineskip}
\caption{Evaluation Benchmarks}
\label{table:benchmarks}
\centering
\vspace*{-0.5\baselineskip}
\resizebox{\linewidth}{!}{ 
\begin{tabular}{cccc}
\toprule
\textbf{\multirow{2}{*}{\makecell[c]{Domain/\\Application}}} &\textbf{\multirow{2}{*}{\makecell[c]{Kernel}}} & \multicolumn{2}{c}{\textbf{Control Flow Feature}} \\
\cmidrule[\heavyrulewidth](lr){3-4}
 & & \textbf{Loop Control} & \textbf{Branch Divergence}\\
\toprule

\multirow{3}{*}{\textbf{\makecell[c]{Machine\\Learning}}} & conv & Imperfect Nested & N/A\\

 & relu & Simple Loop & Two Branches in Loop\\

 & spmv & Imperfect Nested & N/A\\

\midrule

\multirow{3}{*}{\textbf{\makecell[c]{Linear\\Algebra}}} & gemm & Imperfect Nested & N/A \\

 & bicg & Imperfect Nested & N/A \\

 & mvt & Perfect Nested & N/A \\

\midrule

\multirow{5}{*}{\textbf{\makecell[c]{Signal\\Processing}}} & adpcm & Simple Loop & Four Branches in Loop \\

& dtw & Perfect Nested & Four Branches in Loop \\

& jacobi & Imperfect Nested, Serial Loops & N/A \\

 & fft & Imperfect Nested & N/A \\

 & merge-sort & Perfect Nested & Loop in Nested Branches \\

\midrule
\multirow{3}{*}{\textbf{\makecell[c]{Graph\\Algorithm}}} & dijkstra & Imperfect Nested & Three Branches (Two Nested) in Loop\\

& bfs & Imperfect Nested & One Branch in Loop\\
 & floyd & Perfect Nested & Two Branches in Loop\\

\midrule
\multirow{5}{*}{\textbf{\makecell[c]{2-Layer Graph \\Convolutional\\Network\\(GCN)}}} & compress & Perfect Nested & One Branch in Loop\\
& aggregate ($\times$2) & Perfect Nested & N/A \\
& combine & Perfect Nested & N/A \\
& combRelu & Perfect Nested & One Branch in Loop \\
& pooling & Perfect Nested & N/A \\

\midrule
\multirow{6}{*}{\textbf{\makecell[c]{Lower-Upper\\Decomposition\\(LU)}}} & init & Simple Loop & N/A\\
& decompose & Imperfect Nested & N/A\\
& solver0 & Simple Loop & Two Branches in Loop \\
& solver1 & Simple Loop & Two Branches in Loop \\
& invert & Imperfect Nested & One Branch in Loop \\
& determinant & Simple Loop & Nested Branch in Loop \\

\bottomrule
\end{tabular}
}
\vspace*{-1.0\baselineskip}
\end{wraptable}

\textbf{Comparison Methodology.} To fairly compare strategies for bridging the \textbf{\textit{control-dataflow semantic gap}}, we use the specific compiler provided by each framework to generate its specific kernel representation. All kernel representations are then processed by the same mapping algorithm \cite{OpenCGRA-ICCD2020} across all evaluated architectures to isolate the impact of different mapping algorithms, a process that typically completes in a few minutes. All architectures are normalized to a $6\times6$ fabric size to mitigate topological differences in mapping complexity. As typical CGRAs are statically configured, evaluation metrics (e.g., speedup, instructions-per-cycle) are known at compilation time. We can precisely calculate these metrics from the mapping result. We construct cycle-accurate simulators that account for host-CGRA communication for each architecture based on its specifications. We synthesize the RTL implementation of each architecture using Synopsys Design Compiler with the TSMC 22nm ULL library at 800MHz to get area for performance-per-area evaluation (see Sec. \ref{sec8.2:control-eval}).

\begin{wrapfigure}{r}{0.23\textwidth}
  \vspace*{-1.0\baselineskip}
  \centering
  \includegraphics[width=0.23\textwidth]{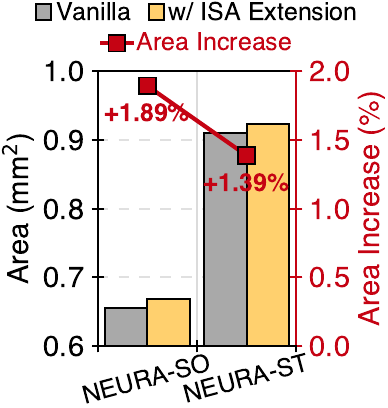}
  \vspace*{-1.3\baselineskip}
  \caption{Area Overhead of NEURA ISA Extensions.}
  \label{fig:area-compare}
  \vspace*{-0.7\baselineskip}
\end{wrapfigure}
\textbf{Evaluation Methodology for NEURA Acceleration Granularity.} NEURA's ability to represent complex control flow as a unified DFG provides the flexibility to choose the acceleration granularity --- ranging from a single inner loop to the entire kernel. Choosing the optimal granularity is complex, as offloading the entire kernel may not always yield optimal performance (e.g., outer loops with very few iterations prevent effective amortization of initialization overheads). While automatically determining this optimal granularity is a research problem beyond this paper's scope, we adopt an iterative methodology for the evaluation. We compile and simulate each possible granularity (from the innermost loop to the entire kernel) and select the one yielding the lowest total kernel execution latency. Critically, only NEURA possesses this flexibility in our evaluation. Although RipTide can also flatten complex kernels into a DFG, its rigid \textbf{\textit{spatial-only}} model and fixed array size constrain it to executing only the innermost loops in our evaluation.

\subsection{Hardware Overhead Analysis}\label{sec8.15:hardware-overhead}
To show the hardware overhead of ISA extensions required by Sec. \ref{sec4.3:interface}, we synthesize $6\times 6$ NEURA-enabled architectures (NEURA-SO and NEURA-ST) and compare them to corresponding vanilla baselines (without ISA extensions). As shown in Fig. \ref{fig:area-compare}, the area overhead is negligible: only $1.89\%$ for NEURA-SO and $1.39\%$ for NEURA-ST. Furthermore, by bundling and routing the predicate bit alongside the existing data payload, this design incurs negligible routing overhead. Consequently, both architectures maintain the vanilla baselines' 800MHz operating frequency without degradation. This low cost demonstrates that the ISA extensions to support NEURA's predicated type system and specialized instructions are highly efficient.

\subsection{NEURA-ST Outperforms SOTA Control Flow Handling Strategies}\label{sec8.2:control-eval}
To demonstrate a high-performance solution for handling control flow on CGRAs, we evaluate NEURA-ST against the three SOTA baselines across performance (Speedup), Instructions Per Cycle (IPC), and area efficiency (Performance Per Area). Results reveal NEURA-ST consistently outperforms all baselines in performance and achieves high improvements in area efficiency.

\textbf{Performance.} As shown in Fig. \ref{fig:speedup-compare}, NEURA-ST achieves geometric mean (geomean) speedups of $2.20\times$ over Marionette, $2.24\times$ over ICED, and $2.42\times$ over RipTide. These gains stem directly from NEURA's ability to generate a single, unified DFG that holistically represents the kernel. This unification eliminates the serialization bottlenecks inherent in Marionette's dedicated controller approach, which sequentially executes BBs and incurs high reconfiguration latency. NEURA's ability to handle hierarchical predicates enables the flattening of complex nested control flow, achieving a higher $2.50\times$ geomean speedup over ICED on benchmarks with hierarchical nested control flows (i.e., \texttt{relu}, \texttt{adpcm}, \texttt{dtw}, \texttt{merge-sort}, \texttt{dijkstra}, \texttt{bfs}, and \texttt{floyd}). This is because ICED can only predicate the inner branch on these benchmarks and reverts to sequential loop iteration execution, resulting in a higher performance gap. The NEURA Dataflow IR is decoupled from specific execution models, allowing it to leverage the flexible \textbf{\textit{spatio-temporal}} execution model to route long data dependencies both temporally (across cycles) and spatially. This yields a geomean speedup of $3.59\times$ over RipTide on benchmarks with long data dependencies (i.e., \texttt{fft}, \texttt{merge-sort}, \texttt{dijkstra}, \texttt{bfs}, \texttt{floyd}), as RipTide's rigid \textbf{\textit{spatial-only}} execution model can only map these dependencies as long, complex spatial routes, increasing the critical path delay.

\begin{figure}[t]
    \centering
    \includegraphics[width=1.0\textwidth]{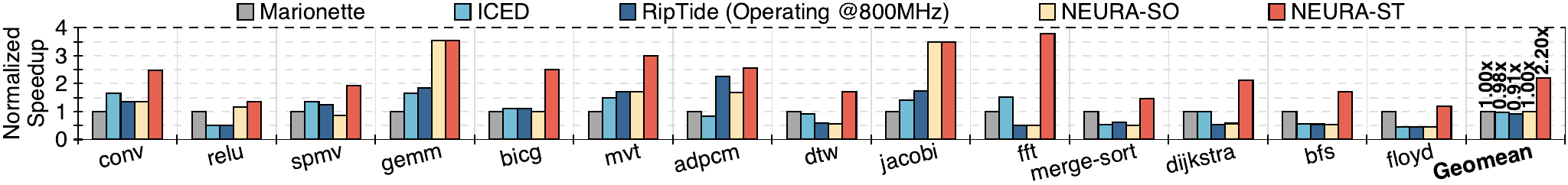}
    \vspace*{-1.4\baselineskip}
    \caption{Performance Comparison --- All results are normalized to the speedup of Marionette. The rightmost group of bars represents the geometric mean (geomean) across all benchmarks.}
    \label{fig:speedup-compare}
\end{figure}

\begin{figure}[t]
    \centering
    \includegraphics[width=1.0\textwidth]{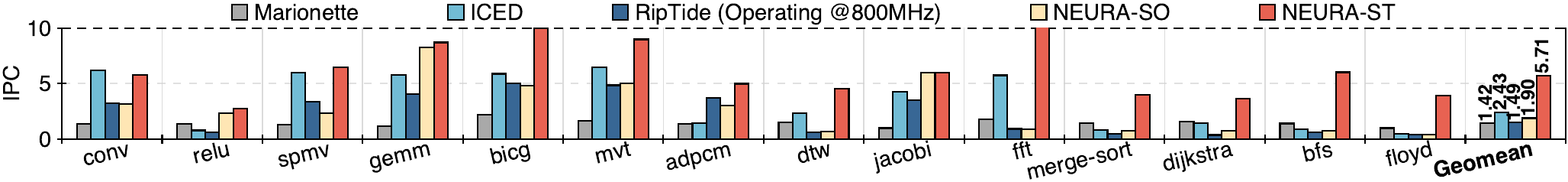}
    \vspace*{-1.4\baselineskip}
    \caption{Instructions Per Cycle (IPC) Comparison --- IPC achieved by NEURA architectures and baselines. A higher IPC indicates more effective exploitation of instruction-level parallelism (ILP).}
    \label{fig:ipc-compare}
\end{figure}

\begin{figure}[t]
    \centering
    \includegraphics[width=1.0\textwidth]{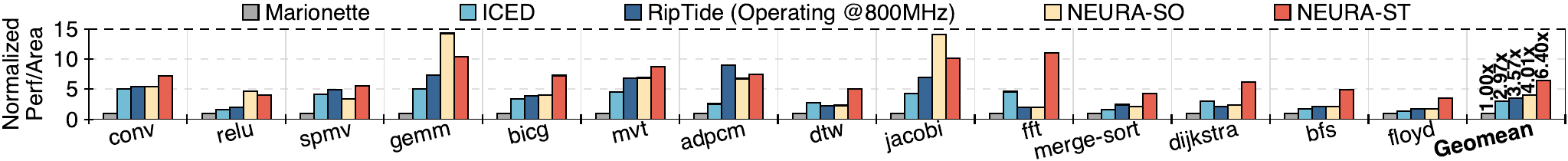}
    \vspace*{-1.4\baselineskip}
    \caption{Performance Per Area (Perf/Area) Comparison --- All results are normalized to the Perf/Area of Marionette. It illustrates the area efficiency of each architecture.}
    \label{fig:ppa-compare}
    \vspace*{-0.65\baselineskip}
\end{figure}

\textbf{Instructions Per Cycle (IPC).} Fig. \ref{fig:ipc-compare} illustrates the IPC, defined as total tile executions divided by total execution latency. NEURA-ST achieves higher geomean IPC over baselines. NEURA's unified DFG fundamentally eliminates the severe control-induced idle cycles and inter-BB reconfiguration stalls inherent in Marionette, resulting in low total execution latency. The comparison with ICED reveals a critical insight. Although ICED shows high IPC on some benchmarks (e.g., \texttt{conv}, \texttt{spmv}), its performance speedup trails NEURA-ST. This paradox is because ICED's kernel IR is not amenable to critical hardware-agnostic optimizations (e.g., data type alignment). This IR limitation leaves it unable to eliminate redundant operations, inflating its total tile executions and harming performance. NEURA supports these optimizations, producing a more efficient IR. Finally, RipTide's IPC collapses on benchmarks with long data dependencies. This further validates our performance analysis --- its rigid \textbf{\textit{spatial-only}} execution model creates long, static routes that drastically inflate the total execution latency without adding proportional computational work, destroying the IPC.


\textbf{Performance Per Area.} Fig. \ref{fig:ppa-compare} shows the performance per area (Perf/Area) comparison normalized to Marionette. NEURA-ST achieves a remarkable $6.40\times$ geomean Perf/Area over Marionette. This exceptional efficiency stems from NEURA's compiler-centric design. NEURA achieves its high performance using only minimal ISA extensions, avoiding the significant area overhead of Marionette's dedicated in-tile control hardware. NEURA-ST's $2.15\times$ geomean Perf/Area over ICED further highlights the value of NEURA Dataflow IR. With comparable hardware overhead (NEURA-ST $0.92\text{mm}^2$, ICED $0.88\text{mm}^2$), NEURA fully exploits the available hardware parallelism by resolving hierarchical predicates, achieving high area efficiency. Finally, NEURA-ST surpasses RipTide by $1.79\times$ by achieving a superior design tradeoff. While RipTide's low-power design is area-efficient, it achieves this by sacrificing performance with its rigid \textbf{\textit{spatial-only}} execution model. NEURA-ST delivers a much higher performance while maintaining excellent area efficiency.


\subsection{NEURA-SO Provides a Competitive Solution for Low-Power CGRAs}\label{sec8.25:low-power}
Sec. \ref{sec8.2:control-eval} demonstrates NEURA's high-performance capabilities with the \textbf{\textit{spatio-temporal}} NEURA-ST architecture. NEURA's retargetability also extends to the low-power domain. To assess its effectiveness in this domain, we evaluate the NEURA Dataflow IR on our \textbf{\textit{spatial-only}} NEURA-SO architecture. We compare its performance and energy against RipTide, the SOTA low-power CGRA, to demonstrate that NEURA provides a competitive and general solution in the low-power domain.


\begin{wrapfigure}{r}{0.6\textwidth}
  \vspace*{-0.3\baselineskip}
  \centering
  \includegraphics[width=0.6\textwidth]{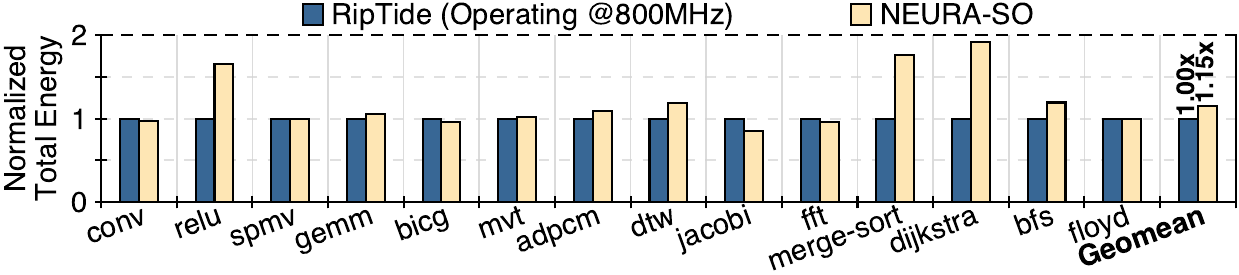}
  \caption{Energy Comparison --- Energy consumption of NEURA-SO normalized to RipTide (operating at 800MHz).}
  \label{fig:energy-compare}
  \vspace*{-0.3\baselineskip}
\end{wrapfigure}

\textbf{Performance, IPC, and Perf/Area.} As shown in Fig.\ref{fig:speedup-compare}, NEURA-SO delivers competitive performance, exceeding RipTide by $10\%$ on geomean. It also achieves geomean improvements of $28\%$ in IPC (see Fig. \ref{fig:ipc-compare}) and $12\%$ in Perf/Area (see Fig. \ref{fig:ppa-compare}) over RipTide. For benchmarks like \texttt{relu}, \texttt{gemm}, and \texttt{jacobi}, NEURA-SO significantly outperforms RipTide in performance. This occurs because RipTide's highly specialized ISA requires regular control patterns to be converted to match its specific steering-control primitives, which lengthens the critical path. Conversely, RipTide performs better on \texttt{spmv}, where its irregular memory access patterns align well 
with RipTide's specialized ISA.

\textbf{Energy Consumption.} Fig.\ref{fig:energy-compare} shows that NEURA-SO's energy consumption is competitive with RipTide across most benchmarks, with a geomean only $1.15\times$ higher. RipTide (operating at 800MHz) demonstrates lower energy consumption on specific benchmarks like \texttt{relu}, \texttt{merge-sort}, and \texttt{dijkstra}. This is attributable to RipTide's specialized \texttt{merge} operation, which efficiently fuses multi-input selection logic from the branches within these benchmarks' loops. In contrast, the general NEURA Dataflow IR, designed for portability, does not presume such specialized operations.


\textbf{Generality vs. Specialization.} The evaluation validates that the general NEURA Dataflow IR can achieve competitive performance and energy consumption against the SOTA low-power CGRA. The observed trade-offs in performance and energy underscore NEURA's design philosophy: relying on a small set of general operations provides strong baseline efficiency and avoids over-specialization. This general design is not a limitation. For scenarios demanding further optimization, NEURA's extensible framework can readily incorporate hardware-specific operations, via pattern rewriting in Sec. \ref{sec6.2:hw-specific}. Our open-source release already extends our IR with RipTide's specialized operations and implements the transformation from our general IR to these specific operations.


\subsection{Impact of NEURA Compiler Optimizations}\label{sec8.3:optimization-expr}

To quantify the impact of NEURA compiler optimizations presented in Sec. \ref{sec6:optimization}, we evaluate their cumulative contributions. The baseline is the unoptimized NEURA Dataflow IR executed on a $6\times 6$ NEURA-ST architecture augmented with specialized FUs supporting the \texttt{neura.load\_indexed} and \texttt{neura.loop\_control} operations (see Sec. \ref{sec6.2:hw-specific}). Fig. \ref{fig:opt-compare} shows the cumulative speedup as each optimization is applied.



\textbf{Hardware-agnostic optimizations} yield significant gains. \textbf{Data type alignment} and \textbf{constant folding} provide a $1.69\times$ geomean speedup over the baseline. This improvement arises because these optimizations eliminate redundant type conversions and avoid the overhead of materializing constants as DFG operations. Folding constants into attributes avoids additional predicate management logic, simplifies the data dependencies within the DFG, and frees up hardware resources.

\begin{wrapfigure}{r}{0.58\textwidth}
  \centering
  \includegraphics[width=0.58\textwidth]{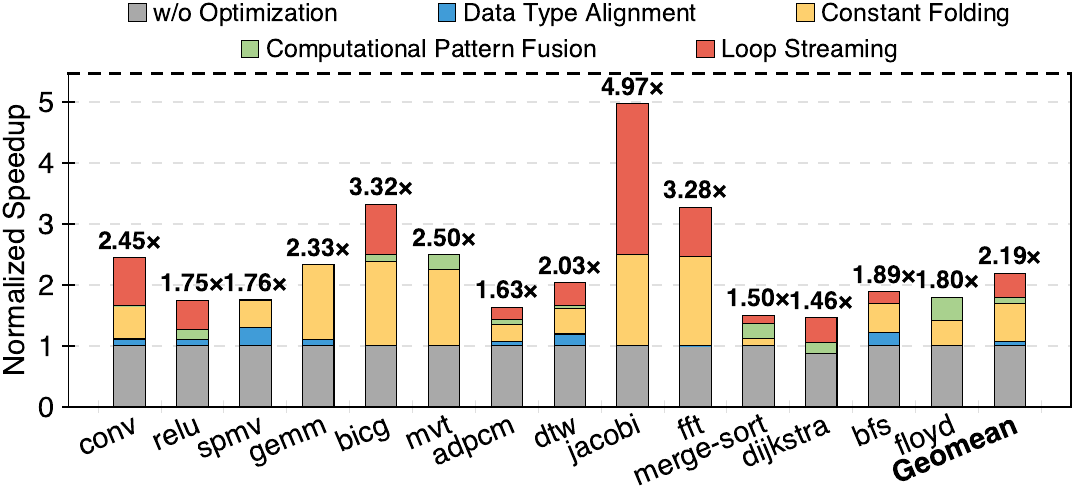}
  \vspace*{-1.2\baselineskip}
  \caption{Impact of NEURA Compiler Optimizations --- Cumulative speedup on an augmented $6\times6$ NEURA-ST normalized to source IR without optimization. Bars show the cumulative speedup after enabling optimizations.}
  \label{fig:opt-compare}
  \vspace*{-0.7\baselineskip}
\end{wrapfigure}

Further performance gains are contributed by \textbf{hardware-specific optimizations}. First, we apply the \textbf{computational pattern fusion}. For this evaluation, we enable the load fusion pattern to fuse address calculation and memory access into a single \texttt{neura.load\_indexed} operation. This brings the cumulative geomean speedup to $1.79\times$ over the baseline. Next, the \textbf{loop streaming optimization} delivers further gains, especially for benchmarks bottlenecked by inter-iteration dependencies in their loop control logic (e.g., \texttt{\texttt{jacobi}}). By fusing the loop control logic into a single \texttt{neura.loop\_control} operation, this optimization breaks the critical inter-iteration dependencies, enabling more aggressive pipelining. The cumulative effect of all optimizations results in a geomean speedup of $2.19\times$ over the baseline, demonstrating the efficacy and importance of NEURA's multi-level optimization strategy.

\subsection{NEURA Dataflow IR is Scalable}\label{sec8.4:scalability}
A key attribute of a versatile CGRA compilation framework is its ability to scale performance across different hardware fabric sizes. We conduct a scalability study to demonstrate that the NEURA Dataflow IR is not a bottleneck as hardware resources increase. Our initial analysis reveals that $6\times6$ NEURA-ST already saturates performance for most benchmarks due to inter-iteration data dependencies rather than resource limitations. Comparing against an even larger fabric would fail to reveal the IR's scalability. Therefore, we compare the performance of the same NEURA Dataflow IR on a smaller $4\times 4$ NEURA-ST against the base $6\times 6$ array. We exclude benchmarks whose performance is already saturated on the $4\times4$ array due to inter-iteration data dependencies rather than resource constraints.

\begin{wrapfigure}{r}{0.5\textwidth}
  \vspace*{-1.0\baselineskip}
  \centering
  \includegraphics[width=0.5\textwidth]{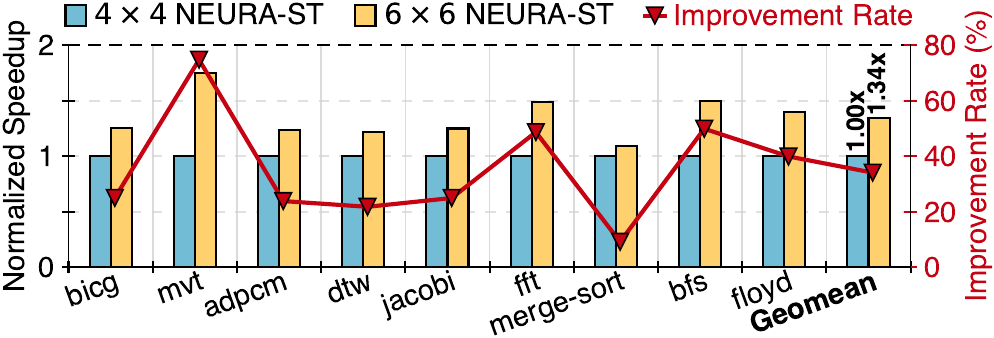}
  \vspace*{-1.3\baselineskip}
  \caption{Scalability Analysis --- Performance of running the same NEURA Dataflow IR on both $4\times4$ and $6\times6$ NEURA-ST architectures.}
  \label{fig:scalability-compare}
  \vspace*{-0.5\baselineskip}
\end{wrapfigure}
As shown in Fig. \ref{fig:scalability-compare}, scaling the architecture from $4\times4$ to $6\times6$ yields performance improvement across all tested benchmarks, resulting in a geomean speedup of $1.34\times$. This performance scaling is not perfectly linear, which is expected due to two inherent factors. First, the performance might saturate on the $4\times4$ array, limiting additional benefits from a larger array. Second, performance scaling is bounded by the theoretical minimum resource II, defined as the instruction count in the DFG divided by the tile count. For example, \texttt{floyd} contains 39 instructions. On a $4\times4$ (16-tile) array, its minimum resource II is $\lceil39/16\rceil=3$. Scaling to a $6\times6$ (36-tile) array only reduces the minimum resource II to $\lceil39/36\rceil=2$. This implies a theoretical maximum speedup of $1.5\times$. Considering these inherent factors, the achieved $1.34\times$ speedup validates the scalability of the NEURA Dataflow IR.

\subsection{NEURA Effectively Accelerates Real-World Applications}\label{sec8.28:application-expr}
To clearly illustrate the performance disparities across different architectures on real-world workloads, Fig. \ref{fig:application-perf} presents the execution time breakdown of the constituent kernels within the GCN and LU Decomposition applications. All results are normalized to the execution time of Marionette.

\textbf{Performance of NEURA-ST.} NEURA-ST achieves the lowest execution time among all evaluated architectures for both applications. Specifically, NEURA-ST achieves geomean speedup of $2.57\times$ and $2.71\times$ over all evaluated baselines (i.e., Marionette, ICED, and RipTide) on GCN and LU Decomposition, respectively. This performance validates the efficacy of NEURA's hierarchical predication, which seamlessly flattens complex nested loops and branches into a unified DFG. By doing so, NEURA-ST effectively manages the control overheads present in the CDFG Strategy (e.g., Marionette) and the loop-handling inefficiencies in the Limited Predication Strategy (e.g., ICED).

\begin{wrapfigure}{r}{0.63\textwidth}
  \vspace*{-0.4\baselineskip}
  \centering
  \includegraphics[width=0.63\textwidth]{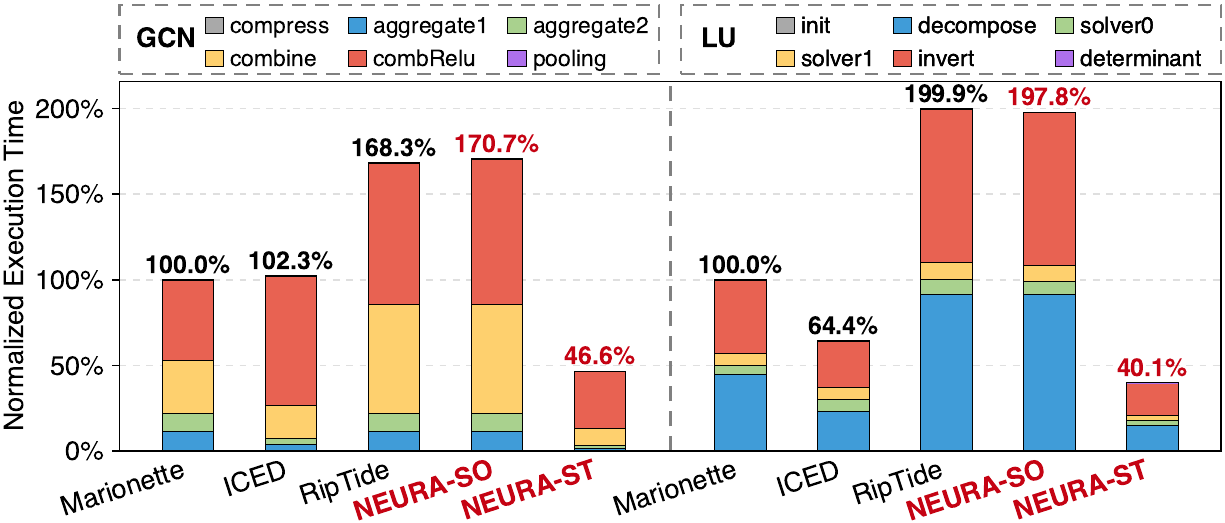}
  \vspace*{-0.8\baselineskip}
  \caption{Application Execution Time Breakdown --- Execution time breakdown of the constituent kernels for GCN and LU Decomposition. All results are normalized to the execution time of Marionette.}
  \label{fig:application-perf}
  \vspace*{-0.4\baselineskip}
\end{wrapfigure}

\textbf{Performance of NEURA-SO.} When targeting the \textit{\textbf{spatial-only}} execution model, NEURA-SO delivers performance comparable to the SOTA low-power CGRA baseline, RipTide, across both applications. While NEURA-SO and RipTide exhibit longer execution time than high-performance architectures like Marionette and NEURA-ST, this reflects an expected trade-off for low-power CGRAs that prioritize extreme energy efficiency over sheer speed. Achieving performance on par with RipTide, coupled with the substantial speedups delivered by NEURA-ST, demonstrates NEURA's capability to retarget applications to different execution models to meet different hardware constraints and performance demands.

%% file: Ch9-Related.tex
\section{Related Works}\label{sec9:related-works}
\textbf{Predicated Execution for CGRAs.} Predicated execution is a well-established technique aiming to eliminate control flow branches for high ILP \cite{Hyperblock-MICRO1992, FullPartialPredicate-1995ISCA, IfConversion-1983POPL, PredicateAnalysis-1996MICRO, PredicatedExecution-1991Tech}, employed across various architectures like CPUs and GPUs \cite{VLIWPredicate-1999ISCA, AMDGPUPredicate-2011CompArchNews, NVIDIAGPUPredicate-2025online}. Recognizing its potential, CGRAs have explored predicated execution, but existing approaches exhibit key limitations. One category of works focuses solely on handling intra-loop branch divergence (i.e., \texttt{if-else}) using predication \cite{BranchAware-2014DAC, AdvancedPredication-2010FPT, TLIA-2016TPDS, TIA-2015ICCAD,TRIPS-2004TACO,SARA-ISCA2021}. However, these solutions do not address the predication of the loop control logic, relying on external hardware to manage loop iterations. Another category attempts to use predication for simple loop control (i.e., single loop) \cite{OpenCGRA-ICCD2020, Morpher-2022WOSET, Plaid-ASPLOS2025, ICED-MICRO2024}. However, these solutions fail when faced with nested control structures (e.g., nested loops, branches in loops). They often revert to only predicating either the innermost loop or inner branch, failing to flatten the encompassing loop structure because they cannot represent the required hierarchical predicate context. Consequently, existing predicated execution solutions for CGRAs lack a mechanism to systematically unify predicate contexts originating from different control levels. NEURA addresses this limitation through its predicated type system and predicate management operations, providing the first mechanism to systematically represent and combine these hierarchical predicate contexts directly within a pure dataflow graph.

\textbf{IRs and Compilation Frameworks for CGRAs.} Compilation frameworks for CGRAs heavily rely on CDFG representations derived from imperative languages \cite{LLVM-CGO2004, MLIR-CGO2021, DFG-CACM1976, PDG-1987TOPLAS}. Existing frameworks typically follow one of three main strategies to manage the CDFG. The first category separates the CFG and DFG execution, mapping only DFGs onto the CGRA while managing the CFG externally via host CPUs or dedicated hardware \cite{SARA-ISCA2021, Marionette-MICRO2023, Fifer-2021MICRO, Stardust-2025CGO, SNAFU-ISCA2021, CGRA-ME-2017ASAP}. While Spatial \cite{Spatial-PLDI2018} can generate efficient DFG configurations for CGRAs, this separation fundamentally limits inter-DFG parallelism. The second category attempts to unify the representation into a single DFG. Steering control techniques \cite{SteerControl-PL1991, SpatialComp-2004ASPLOS, Ripple-PLDI2025, RipTide-MICRO2022, NUPEA-ISCA2025, DataflowProcessor-ISCA1974, WaveScalar-2003MICRO, MITTaged-1990TC} achieve full flattening and are widely adopted in the \textbf{\textit{spatial-only}} execution model. However, as discussed in Sec. \ref{subsec2.2:controlflow}, steering control is difficult to adapt to \textbf{\textit{spatio-temporal}} execution, sacrificing performance and architectural generality for complex control flows. The third category employs predicated execution \cite{ICED-MICRO2024, Plaid-ASPLOS2025, Morpher-2022WOSET, OpenCGRA-ICCD2020, ML-CGRA-2023DAC,SARA-ISCA2021,TRIPS-2004TACO}. They can handle simple branches, but fail to represent nested control flows, preventing the complete conversion of CDFGs. Consequently, existing frameworks lack a unified, pure dataflow IR capable of representing hierarchical control flow while remaining decoupled from specific execution models. NEURA addresses this critical gap with its NEURA Dataflow IR, which leverages a predicated type system to provide the first such unified representation, inherently supporting both \textbf{\textit{spatial-only}} and \textbf{\textit{spatio-temporal}} execution.

%% file: ref.bib
@String{Computing = "Computing" }

@String{Computer = "{IEEE} Computer" }

@String{Springer = "Springer-Verlag" }

@inproceedings{dark-silicon-ISCA2011,
author = {Esmaeilzadeh, Hadi and Blem, Emily and St. Amant, Renee and Sankaralingam, Karthikeyan and Burger, Doug},
title = {Dark silicon and the end of multicore scaling},
year = {2011},
isbn = {9781450304726},
publisher = {Association for Computing Machinery},
address = {New York, NY, USA},
url = {https://doi.org/10.1145/2000064.2000108},
doi = {10.1145/2000064.2000108},
booktitle = {Proceedings of the 38th Annual International Symposium on Computer Architecture},
pages = {365–376},
numpages = {12},
location = {San Jose, California, USA},
series = {ISCA '11}
}

@article{DSA-accelerators-CommunACM2020,
author = {Dally, William J. and Turakhia, Yatish and Han, Song},
title = {Domain-specific hardware accelerators},
year = {2020},
issue_date = {July 2020},
publisher = {Association for Computing Machinery},
address = {New York, NY, USA},
volume = {63},
number = {7},
issn = {0001-0782},
url = {https://doi.org/10.1145/3361682},
doi = {10.1145/3361682},
journal = {Commun. ACM},
month = jun,
pages = {48–57},
numpages = {10}
}

@ARTICLE{dark-silicon-server-Micro2011,
  author={Hardavellas, Nikos and Ferdman, Michael and Falsafi, Babak and Ailamaki, Anastasia},
  journal={IEEE Micro}, 
  title={Toward Dark Silicon in Servers}, 
  year={2011},
  volume={31},
  number={4},
  pages={6-15},
  doi={10.1109/MM.2011.77}}

@inproceedings{HiPER-ISCA2025,
author = {Ting, Justin and Kim, Minsik and Zhu, Junkang and Sheng, Haotian and Zhang, Zhengya},
title = {HiPER: Hierarchically-Composed Processing for Efficient Robot Learning-Based Control},
year = {2025},
isbn = {9798400712616},
publisher = {Association for Computing Machinery},
address = {New York, NY, USA},
url = {https://doi.org/10.1145/3695053.3731097},
doi = {10.1145/3695053.3731097},
booktitle = {Proceedings of the 52nd Annual International Symposium on Computer Architecture},
pages = {313–326},
numpages = {14},
location = {
},
series = {ISCA '25}
}

@inproceedings{Dadu-RBD-MICRO2023,
author = {Yang, Yuxin and Chen, Xiaoming and Han, Yinhe},
title = {Dadu-RBD: Robot Rigid Body Dynamics Accelerator with Multifunctional Pipelines},
year = {2023},
isbn = {9798400703294},
publisher = {Association for Computing Machinery},
address = {New York, NY, USA},
url = {https://doi.org/10.1145/3613424.3614298},
doi = {10.1145/3613424.3614298},
booktitle = {Proceedings of the 56th Annual IEEE/ACM International Symposium on Microarchitecture},
pages = {297–309},
numpages = {13},
location = {Toronto, ON, Canada},
series = {MICRO '23}
}

@INPROCEEDINGS{Cambricon-X-MICRO2016,
  author={Zhang, Shijin and Du, Zidong and Zhang, Lei and Lan, Huiying and Liu, Shaoli and Li, Ling and Guo, Qi and Chen, Tianshi and Chen, Yunji},
  booktitle={2016 49th Annual IEEE/ACM International Symposium on Microarchitecture (MICRO)}, 
  title={Cambricon-X: An accelerator for sparse neural networks}, 
  year={2016},
  volume={},
  number={},
  pages={1-12},
  doi={10.1109/MICRO.2016.7783723}}

@INPROCEEDINGS{dadiannao-MICRO2014,
  author={Chen, Yunji and Luo, Tao and Liu, Shaoli and Zhang, Shijin and He, Liqiang and Wang, Jia and Li, Ling and Chen, Tianshi and Xu, Zhiwei and Sun, Ninghui and Temam, Olivier},
  booktitle={2014 47th Annual IEEE/ACM International Symposium on Microarchitecture}, 
  title={DaDianNao: A Machine-Learning Supercomputer}, 
  year={2014},
  volume={},
  number={},
  pages={609-622},
  doi={10.1109/MICRO.2014.58}}

@ARTICLE{Amber-JSSC2024,
  author={Feng, Kathleen and Kong, Taeyoung and Koul, Kalhan and Melchert, Jackson and Carsello, Alex and Liu, Qiaoyi and Nyengele, Gedeon and Strange, Maxwell and Zhang, Keyi and Nayak, Ankita and Setter, Jeff and Thomas, James and Sreedhar, Kavya and Chen, Po-Han and Bhagdikar, Nikhil and Myers, Zach A. and D’Agostino, Brandon and Joshi, Pranil and Richardson, Stephen and Torng, Christopher and Horowitz, Mark and Raina, Priyanka},
  journal={IEEE Journal of Solid-State Circuits}, 
  title={Amber: A 16-nm System-on-Chip With a Coarse- Grained Reconfigurable Array for Flexible Acceleration of Dense Linear Algebra}, 
  year={2024},
  volume={59},
  number={3},
  pages={947-959},
  doi={10.1109/JSSC.2023.3313116}}

@INPROCEEDINGS{RipTide-MICRO2022,
  author={Gobieski, Graham and Ghosh, Souradip and Heule, Marijn and Mowry, Todd and Nowatzki, Tony and Beckmann, Nathan and Lucia, Brandon},
  booktitle={2022 55th IEEE/ACM International Symposium on Microarchitecture (MICRO)}, 
  title={RipTide: A Programmable, Energy-Minimal Dataflow Compiler and Architecture}, 
  year={2022},
  volume={},
  number={},
  pages={546-564},
  doi={10.1109/MICRO56248.2022.00046}}

@INPROCEEDINGS{REVEL-HPCA2020,
  author={Weng, Jian and Liu, Sihao and Wang, Zhengrong and Dadu, Vidushi and Nowatzki, Tony},
  booktitle={2020 IEEE International Symposium on High Performance Computer Architecture (HPCA)}, 
  title={A Hybrid Systolic-Dataflow Architecture for Inductive Matrix Algorithms}, 
  year={2020},
  volume={},
  number={},
  pages={703-716},
  doi={10.1109/HPCA47549.2020.00063}}

@INPROCEEDINGS{ICED-MICRO2024,
  author={Tan, Cheng and Jiang, Miaomiao and Patil, Deepak and Ou, Yanghui and Li, Zhaoying and Ju, Lei and Mitra, Tulika and Park, Hyunchul and Tumeo, Antonino and Zhang, Jeff},
  booktitle={2024 57th IEEE/ACM International Symposium on Microarchitecture (MICRO)}, 
  title={ICED: An Integrated CGRA Framework Enabling DVFS-Aware Acceleration}, 
  year={2024},
  volume={},
  number={},
  pages={1338-1352},
  doi={10.1109/MICRO61859.2024.00099}}

@article{Ripple-PLDI2025,
author = {Ghosh, Souradip and Shi, Yufei and Lucia, Brandon and Beckmann, Nathan},
title = {Ripple: Asynchronous Programming for Spatial Dataflow Architectures},
year = {2025},
issue_date = {June 2025},
publisher = {Association for Computing Machinery},
address = {New York, NY, USA},
volume = {9},
number = {PLDI},
url = {https://doi.org/10.1145/3729256},
doi = {10.1145/3729256},
journal = {Proc. ACM Program. Lang.},
month = jun,
articleno = {157},
numpages = {28}
}

@inproceedings{NUPEA-ISCA2025,
author = {Ghosh, Souradip and Gobieski, Graham and Zhang, Keyi and Lucia, Brandon and Beckmann, Nathan and Nowatzki, Tony},
title = {NUPEA: Optimizing Critical Loads on Spatial Dataflow Architectures via Non-Uniform Processing-Element Access},
year = {2025},
isbn = {9798400712616},
publisher = {Association for Computing Machinery},
address = {New York, NY, USA},
url = {https://doi.org/10.1145/3695053.3731061},
doi = {10.1145/3695053.3731061},
booktitle = {Proceedings of the 52nd Annual International Symposium on Computer Architecture},
pages = {1627–1640},
numpages = {14},
location = {
},
series = {ISCA '25}
}

@article{Plasticine-ISCA2017,
author = {Prabhakar, Raghu and Zhang, Yaqi and Koeplinger, David and Feldman, Matt and Zhao, Tian and Hadjis, Stefan and Pedram, Ardavan and Kozyrakis, Christos and Olukotun, Kunle},
title = {Plasticine: A Reconfigurable Architecture For Parallel Paterns},
year = {2017},
issue_date = {May 2017},
publisher = {Association for Computing Machinery},
address = {New York, NY, USA},
volume = {45},
number = {2},
issn = {0163-5964},
url = {https://doi.org/10.1145/3140659.3080256},
doi = {10.1145/3140659.3080256},
journal = {SIGARCH Comput. Archit. News},
month = jun,
pages = {389–402},
numpages = {14}
}

@INPROCEEDINGS{UE-CGRA-HPCA2021,
  author={Torng, Christopher and Pan, Peitian and Ou, Yanghui and Tan, Cheng and Batten, Christopher},
  booktitle={2021 IEEE International Symposium on High-Performance Computer Architecture (HPCA)}, 
  title={Ultra-Elastic CGRAs for Irregular Loop Specialization}, 
  year={2021},
  volume={},
  number={},
  pages={412-425},
  doi={10.1109/HPCA51647.2021.00042}}

@inproceedings{ModuloScheduling-MICRO1994,
author = {Rau, B. Ramakrishna},
title = {Iterative modulo scheduling: an algorithm for software pipelining loops},
year = {1994},
isbn = {0897917073},
publisher = {Association for Computing Machinery},
address = {New York, NY, USA},
url = {https://doi.org/10.1145/192724.192731},
doi = {10.1145/192724.192731},
booktitle = {Proceedings of the 27th Annual International Symposium on Microarchitecture},
pages = {63–74},
numpages = {12},
location = {San Jose, California, USA},
series = {MICRO 27}
}

@INPROCEEDINGS{SARA-ISCA2021,
  author={Zhang, Yaqi and Zhang, Nathan and Zhao, Tian and Vilim, Matt and Shahbaz, Muhammad and Olukotun, Kunle},
  booktitle={2021 ACM/IEEE 48th Annual International Symposium on Computer Architecture (ISCA)}, 
  title={SARA: Scaling a Reconfigurable Dataflow Accelerator}, 
  year={2021},
  volume={},
  number={},
  pages={1041-1054},
  doi={10.1109/ISCA52012.2021.00085}}

@inproceedings{Marionette-MICRO2023,
author = {Deng, Jinyi and Tang, Xinru and Zhang, Jiahao and Li, Yuxuan and Zhang, Linyun and Han, Boxiao and He, Hongjun and Tu, Fengbin and Liu, Leibo and Wei, Shaojun and Hu, Yang and Yin, Shouyi},
title = {Towards Efficient Control Flow Handling in Spatial Architecture via Architecting the Control Flow Plane},
year = {2023},
isbn = {9798400703294},
publisher = {Association for Computing Machinery},
address = {New York, NY, USA},
url = {https://doi.org/10.1145/3613424.3614246},
doi = {10.1145/3613424.3614246},
booktitle = {Proceedings of the 56th Annual IEEE/ACM International Symposium on Microarchitecture},
pages = {1395–1408},
numpages = {14},
location = {Toronto, ON, Canada},
series = {MICRO '23}
}

@article{SteerControl-PL1991,
author = {Cytron, Ron and Ferrante, Jeanne and Rosen, Barry K. and Wegman, Mark N. and Zadeck, F. Kenneth},
title = {Efficiently computing static single assignment form and the control dependence graph},
year = {1991},
issue_date = {Oct. 1991},
publisher = {Association for Computing Machinery},
address = {New York, NY, USA},
volume = {13},
number = {4},
issn = {0164-0925},
url = {https://doi.org/10.1145/115372.115320},
doi = {10.1145/115372.115320},
journal = {ACM Trans. Program. Lang. Syst.},
month = oct,
pages = {451–490},
numpages = {40}
}

@inproceedings{Hyperblock-MICRO1992,
author = {Mahlke, Scott A. and Lin, David C. and Chen, William Y. and Hank, Richard E. and Bringmann, Roger A.},
title = {Effective compiler support for predicated execution using the hyperblock},
year = {1992},
isbn = {0818631759},
publisher = {IEEE Computer Society Press},
address = {Washington, DC, USA},
booktitle = {Proceedings of the 25th Annual International Symposium on Microarchitecture},
pages = {45–54},
numpages = {10},
location = {Portland, Oregon, USA},
series = {MICRO 25}
}

@article{DFG-CACM1976,
author = {Allen, F. E. and Cocke, J.},
title = {A program data flow analysis procedure},
year = {1976},
issue_date = {March 1976},
publisher = {Association for Computing Machinery},
address = {New York, NY, USA},
volume = {19},
number = {3},
issn = {0001-0782},
url = {https://doi.org/10.1145/360018.360025},
doi = {10.1145/360018.360025},
journal = {Commun. ACM},
month = mar,
pages = {137},
numpages = {11}
}

@INPROCEEDINGS{DRESC-FPT2002,
  author={Bingfeng Mei and Vernalde, S. and Verkest, D. and De Man, H. and Lauwereins, R.},
  booktitle={2002 IEEE International Conference on Field-Programmable Technology, 2002. (FPT). Proceedings.}, 
  title={DRESC: a retargetable compiler for coarse-grained reconfigurable architectures}, 
  year={2002},
  volume={},
  number={},
  pages={166-173},
  doi={10.1109/FPT.2002.1188678}}

@INPROCEEDINGS{SNAFU-ISCA2021,
  author={Gobieski, Graham and Atli, Ahmet Oguz and Mai, Kenneth and Lucia, Brandon and Beckmann, Nathan},
  booktitle={2021 ACM/IEEE 48th Annual International Symposium on Computer Architecture (ISCA)}, 
  title={Snafu: An Ultra-Low-Power, Energy-Minimal CGRA-Generation Framework and Architecture}, 
  year={2021},
  volume={},
  number={},
  pages={1027-1040},
  doi={10.1109/ISCA52012.2021.00084}}

@ARTICLE{DySER-IEEEMicro2012,
  author={Govindaraju, Venkatraman and Ho, Chen-Han and Nowatzki, Tony and Chhugani, Jatin and Satish, Nadathur and Sankaralingam, Karthikeyan and Kim, Changkyu},
  journal={IEEE Micro}, 
  title={DySER: Unifying Functionality and Parallelism Specialization for Energy-Efficient Computing}, 
  year={2012},
  volume={32},
  number={5},
  pages={38-51},
  doi={10.1109/MM.2012.51}}

@inproceedings{Plaid-ASPLOS2025,
author = {Li, Zhaoying and Dangi, Pranav and Yin, Chenyang and Bandara, Thilini Kaushalya and Juneja, Rohan and Tan, Cheng and Bai, Zhenyu and Mitra, Tulika},
title = {Enhancing CGRA Efficiency Through Aligned Compute and Communication Provisioning},
year = {2025},
isbn = {9798400706981},
publisher = {Association for Computing Machinery},
address = {New York, NY, USA},
url = {https://doi.org/10.1145/3669940.3707230},
doi = {10.1145/3669940.3707230},
booktitle = {Proceedings of the 30th ACM International Conference on Architectural Support for Programming Languages and Operating Systems, Volume 1},
pages = {410–425},
numpages = {16},
location = {Rotterdam, Netherlands},
series = {ASPLOS '25}
}

@INPROCEEDINGS{DRIPS-HPCA2022,
  author={Tan, Cheng and Agostini, Nicolas Bohm and Geng, Tong and Xie, Chenhao and Li, Jiajia and Li, Ang and Barker, Kevin J. and Tumeo, Antonino},
  booktitle={2022 IEEE International Symposium on High-Performance Computer Architecture (HPCA)}, 
  title={DRIPS: Dynamic Rebalancing of Pipelined Streaming Applications on CGRAs}, 
  year={2022},
  volume={},
  number={},
  pages={304-316},
  doi={10.1109/HPCA53966.2022.00030}}

@INPROCEEDINGS{HyCUBE-DAC2017,
  author={Karunaratne, Manupa and Mohite, Aditi Kulkarni and Mitra, Tulika and Peh, Li-Shiuan},
  booktitle={2017 54th ACM/EDAC/IEEE Design Automation Conference (DAC)}, 
  title={HyCUBE: A CGRA with reconfigurable single-cycle multi-hop interconnect}, 
  year={2017},
  volume={},
  number={},
  pages={1-6},
  doi={10.1145/3061639.3062262}}

@inproceedings{ADRES-FPL2003,
  title={ADRES: An architecture with tightly coupled VLIW processor and coarse-grained reconfigurable matrix},
  author={Mei, Bingfeng and Vernalde, Serge and Verkest, Diederik and De Man, Hugo and Lauwereins, Rudy},
  booktitle={International conference on field programmable logic and applications},
  pages={61--70},
  year={2003},
  organization={Springer}
}

@INPROCEEDINGS{OpenCGRA-ICCD2020,
  author={Tan, Cheng and Xie, Chenhao and Li, Ang and Barker, Kevin J. and Tumeo, Antonino},
  booktitle={2020 IEEE 38th International Conference on Computer Design (ICCD)}, 
  title={OpenCGRA: An Open-Source Unified Framework for Modeling, Testing, and Evaluating CGRAs}, 
  year={2020},
  volume={},
  number={},
  pages={381-388},
  doi={10.1109/ICCD50377.2020.00070}}

@INPROCEEDINGS{CGRA-ME-2017ASAP,
  author={Chin, S. Alexander and Sakamoto, Noriaki and Rui, Allan and Zhao, Jim and Kim, Jin Hee and Hara-Azumi, Yuko and Anderson, Jason},
  booktitle={2017 IEEE 28th International Conference on Application-specific Systems, Architectures and Processors (ASAP)}, 
  title={CGRA-ME: A unified framework for CGRA modelling and exploration}, 
  year={2017},
  volume={},
  number={},
  pages={184-189},
  doi={10.1109/ASAP.2017.7995277}}

@INPROCEEDINGS{LLVM-CGO2004,
  author={Lattner, C. and Adve, V.},
  booktitle={International Symposium on Code Generation and Optimization, 2004. CGO 2004.}, 
  title={LLVM: a compilation framework for lifelong program analysis \& transformation}, 
  year={2004},
  volume={},
  number={},
  pages={75-86},
  doi={10.1109/CGO.2004.1281665}}

@inproceedings{DataflowProcessor-ISCA1974,
author = {Dennis, Jack B. and Misunas, David P.},
title = {A preliminary architecture for a basic data-flow processor},
year = {1974},
isbn = {9781450373661},
publisher = {Association for Computing Machinery},
address = {New York, NY, USA},
url = {https://doi.org/10.1145/642089.642111},
doi = {10.1145/642089.642111},
booktitle = {Proceedings of the 2nd Annual Symposium on Computer Architecture},
pages = {126–132},
numpages = {7},
series = {ISCA '75}
}

@inproceedings{Spatial-PLDI2018,
author = {Koeplinger, David and Feldman, Matthew and Prabhakar, Raghu and Zhang, Yaqi and Hadjis, Stefan and Fiszel, Ruben and Zhao, Tian and Nardi, Luigi and Pedram, Ardavan and Kozyrakis, Christos and Olukotun, Kunle},
title = {Spatial: a language and compiler for application accelerators},
year = {2018},
isbn = {9781450356985},
publisher = {Association for Computing Machinery},
address = {New York, NY, USA},
url = {https://doi.org/10.1145/3192366.3192379},
doi = {10.1145/3192366.3192379},
booktitle = {Proceedings of the 39th ACM SIGPLAN Conference on Programming Language Design and Implementation},
pages = {296–311},
numpages = {16},
location = {Philadelphia, PA, USA},
series = {PLDI 2018}
}

@INPROCEEDINGS{MLIR-CGO2021,
  author={Lattner, Chris and Amini, Mehdi and Bondhugula, Uday and Cohen, Albert and Davis, Andy and Pienaar, Jacques and Riddle, River and Shpeisman, Tatiana and Vasilache, Nicolas and Zinenko, Oleksandr},
  booktitle={2021 IEEE/ACM International Symposium on Code Generation and Optimization (CGO)}, 
  title={MLIR: Scaling Compiler Infrastructure for Domain Specific Computation}, 
  year={2021},
  volume={},
  number={},
  pages={2-14},
  doi={10.1109/CGO51591.2021.9370308}}

@misc{LLVMDialect-2025Online,
  title        = {'llvm' Dialect},
  author       = {'llvm' Dialect},
  year         = 2025,
  url          = {https://mlir.llvm.org/docs/Dialects/LLVM/#},
  note         = {Accessed on Sept 15, 2025}
}

@misc{arithDialect-2025Online,
  title        = {'arith' Dialect},
  author       = {'arith' Dialect},
  year         = 2025,
  url          = {https://mlir.llvm.org/docs/Dialects/ArithOps/},
  note         = {Accessed on Sept 15, 2025}
}

@misc{LinalgDialect-2025Online,
  title        = {'linalg' Dialect},
  author       = {'linalg' Dialect},
  year         = 2025,
  url          = {https://mlir.llvm.org/docs/Dialects/Linalg/},
  note         = {Accessed on Nov 10, 2025}
}

@misc{TensorDialect-2025Online,
  title        = {'tensor' Dialect},
  author       = {'tensor' Dialect},
  year         = 2025,
  url          = {https://mlir.llvm.org/docs/Dialects/TensorOps/},
  note         = {Accessed on Nov 10, 2025}
}

@misc{memrefDialect-2025Online,
  title        = {'memref' Dialect},
  author       = {'memref' Dialect},
  year         = 2025,
  url          = {https://mlir.llvm.org/docs/Dialects/MemRef/},
  note         = {Accessed on Sept 20, 2025}
}

@INPROCEEDINGS{Polygeist-2021PACT,
  author={Moses, William S. and Chelini, Lorenzo and Zhao, Ruizhe and Zinenko, Oleksandr},
  booktitle={2021 30th International Conference on Parallel Architectures and Compilation Techniques (PACT)}, 
  title={Polygeist: Raising C to Polyhedral MLIR}, 
  year={2021},
  volume={},
  number={},
  pages={45-59},
  doi={10.1109/PACT52795.2021.00011}}

@misc{Torch-MLIR-2025Online,
  title        = {Torch-MLIR},
  author       = {Torch-MLIR},
  year         = 2025,
  url          = {https://github.com/llvm/torch-mlir},
  note         = {Accessed on Sept 15, 2025}
}

@misc{Clang-2025Online,
  title        = {Clang: a C language family frontend for LLVM},
  author       = {Clang},
  year         = 2025,
  url          = {https://clang.llvm.org},
  note         = {Accessed on Sept 15, 2025}
}

@misc{LLVMConstant-2025Online,
  title        = {'llvm' Dialect Constant Operation},
  author       = {'llvm' Dialect},
  year         = 2025,
  url          = {https://mlir.llvm.org/docs/Dialects/LLVM/#llvmmlirconstant-llvmconstantop},
  note         = {Accessed on Oct 21, 2025}
}

@INPROCEEDINGS{MachSuite-2014IISWC,
  author={Reagen, Brandon and Adolf, Robert and Shao, Yakun Sophia and Wei, Gu-Yeon and Brooks, David},
  booktitle={2014 IEEE International Symposium on Workload Characterization (IISWC)}, 
  title={MachSuite: Benchmarks for accelerator design and customized architectures}, 
  year={2014},
  volume={},
  number={},
  pages={110-119},
  doi={10.1109/IISWC.2014.6983050}}

@inproceedings{PolyBench-2014IMPACT,
  title={Understanding polybench/c 3.2 kernels},
  author={Yuki, Tomofumi},
  booktitle={International workshop on polyhedral compilation techniques (IMPACT)},
  pages={1--5},
  year={2014}
}

@INPROCEEDINGS{VecPAC-2023ICCAD,
  author={Tan, Cheng and Patil, Deepak and Tumeo, Antonino and Weisz, Gabriel and Reinhardt, Steve and Zhang, Jeff},
  booktitle={2023 IEEE/ACM International Conference on Computer Aided Design (ICCAD)}, 
  title={VecPAC: A Vectorizable and Precision-Aware CGRA}, 
  year={2023},
  volume={},
  number={},
  pages={1-9},
  doi={10.1109/ICCAD57390.2023.10323910}}

@INPROCEEDINGS{ML-CGRA-2023DAC,
  author={Luo, Yixuan and Tan, Cheng and Agostini, Nicolas Bohm and Li, Ang and Tumeo, Antonino and Dave, Nirav and Geng, Tong},
  booktitle={2023 60th ACM/IEEE Design Automation Conference (DAC)}, 
  title={ML-CGRA: An Integrated Compilation Framework to Enable Efficient Machine Learning Acceleration on CGRAs}, 
  year={2023},
  volume={},
  number={},
  pages={1-6},
  doi={10.1109/DAC56929.2023.10247873}}

@inproceedings{PyTorch-2024ASPLOS,
author = {Ansel, Jason and Yang, Edward and He, Horace and Gimelshein, Natalia and Jain, Animesh and Voznesensky, Michael and Bao, Bin and Bell, Peter and Berard, David and Burovski, Evgeni and Chauhan, Geeta and Chourdia, Anjali and Constable, Will and Desmaison, Alban and DeVito, Zachary and Ellison, Elias and Feng, Will and Gong, Jiong and Gschwind, Michael and Hirsh, Brian and Huang, Sherlock and Kalambarkar, Kshiteej and Kirsch, Laurent and Lazos, Michael and Lezcano, Mario and Liang, Yanbo and Liang, Jason and Lu, Yinghai and Luk, C. K. and Maher, Bert and Pan, Yunjie and Puhrsch, Christian and Reso, Matthias and Saroufim, Mark and Siraichi, Marcos Yukio and Suk, Helen and Zhang, Shunting and Suo, Michael and Tillet, Phil and Zhao, Xu and Wang, Eikan and Zhou, Keren and Zou, Richard and Wang, Xiaodong and Mathews, Ajit and Wen, William and Chanan, Gregory and Wu, Peng and Chintala, Soumith},
title = {PyTorch 2: Faster Machine Learning Through Dynamic Python Bytecode Transformation and Graph Compilation},
year = {2024},
isbn = {9798400703850},
publisher = {Association for Computing Machinery},
address = {New York, NY, USA},
url = {https://doi.org/10.1145/3620665.3640366},
doi = {10.1145/3620665.3640366},
booktitle = {Proceedings of the 29th ACM International Conference on Architectural Support for Programming Languages and Operating Systems, Volume 2},
pages = {929–947},
numpages = {19},
location = {La Jolla, CA, USA},
series = {ASPLOS '24}
}

@ARTICLE{MITTaged-1990TC,
  author={Arvind and Nikhil, R.S.},
  journal={IEEE Transactions on Computers}, 
  title={Executing a program on the MIT tagged-token dataflow architecture}, 
  year={1990},
  volume={39},
  number={3},
  pages={300-318},
  doi={10.1109/12.48862}}

@inproceedings{Morpher-2022WOSET,
  title={Morpher: An open-source integrated compilation and simulation framework for cgra},
  author={Wijerathne, Dhananjaya and Li, Zhaoying and Karunaratne, Manupa and Peh, Li-Shiuan and Mitra, Tulika},
  booktitle={Fifth Workshop on Open-Source EDA Technology (WOSET)},
  year={2022}
}

@INPROCEEDINGS{BranchAware-2014DAC,
  author={Hamzeh, Mahdi and Shrivastava, Aviral and Vrudhula, Sarma},
  booktitle={2014 51st ACM/EDAC/IEEE Design Automation Conference (DAC)}, 
  title={Branch-aware loop mapping on CGRAs}, 
  year={2014},
  volume={},
  number={},
  pages={1-6},
  doi={}}

@article{PDG-1987TOPLAS,
  title={The program dependence graph and its use in optimization},
  author={Ferrante, Jeanne and Ottenstein, Karl J and Warren, Joe D},
  journal={ACM Transactions on Programming Languages and Systems (TOPLAS)},
  volume={9},
  number={3},
  pages={319--349},
  year={1987},
  publisher={ACM New York, NY, USA}
}

@inproceedings{Fifer-2021MICRO,
author = {Nguyen, Quan M. and Sanchez, Daniel},
title = {Fifer: Practical Acceleration of Irregular Applications on Reconfigurable Architectures},
year = {2021},
isbn = {9781450385572},
publisher = {Association for Computing Machinery},
address = {New York, NY, USA},
url = {https://doi.org/10.1145/3466752.3480048},
doi = {10.1145/3466752.3480048},
booktitle = {MICRO-54: 54th Annual IEEE/ACM International Symposium on Microarchitecture},
pages = {1064–1077},
numpages = {14},
location = {Virtual Event, Greece},
series = {MICRO '21}
}

@inproceedings{Stardust-2025CGO,
author = {Hsu, Olivia and Rucker, Alexander and Zhao, Tian and Desai, Varun and Olukotun, Kunle and Kjolstad, Fredrik},
title = {Stardust: Compiling Sparse Tensor Algebra to a Reconfigurable Dataflow Architecture},
year = {2025},
isbn = {9798400712753},
publisher = {Association for Computing Machinery},
address = {New York, NY, USA},
url = {https://doi.org/10.1145/3696443.3708918},
doi = {10.1145/3696443.3708918},
booktitle = {Proceedings of the 23rd ACM/IEEE International Symposium on Code Generation and Optimization},
pages = {628–643},
numpages = {16},
location = {Las Vegas, NV, USA},
series = {CGO '25}
}

@inproceedings{SpatialComp-2004ASPLOS,
  title={Spatial computation},
  author={Budiu, Mihai and Venkataramani, Girish and Chelcea, Tiberiu and Goldstein, Seth Copen},
  booktitle={Proceedings of the 11th international conference on Architectural support for programming languages and operating systems},
  pages={14--26},
  year={2004}
}

@inproceedings{WaveScalar-2003MICRO,
  title={WaveScalar},
  author={Swanson, Steven and Michelson, Ken and Schwerin, Andrew and Oskin, Mark},
  booktitle={Proceedings. 36th Annual IEEE/ACM International Symposium on Microarchitecture, 2003. MICRO-36.},
  pages={291--302},
  year={2003},
  organization={IEEE}
}

@inproceedings{IfConversion-1983POPL,
author = {Allen, J. R. and Kennedy, Ken and Porterfield, Carrie and Warren, Joe},
title = {Conversion of control dependence to data dependence},
year = {1983},
isbn = {0897910907},
publisher = {Association for Computing Machinery},
address = {New York, NY, USA},
url = {https://doi.org/10.1145/567067.567085},
doi = {10.1145/567067.567085},
booktitle = {Proceedings of the 10th ACM SIGACT-SIGPLAN Symposium on Principles of Programming Languages},
pages = {177–189},
numpages = {13},
location = {Austin, Texas},
series = {POPL '83}
}

@inproceedings{FullPartialPredicate-1995ISCA,
  title={A comparison of full and partial predicated execution support for ILP processors},
  author={Mahlke, Scott A and Hank, Richard E and McCormick, James E and August, David I and Hwu, Wen-Mei W},
  booktitle={Proceedings of the 22nd annual international symposium on Computer architecture},
  pages={138--150},
  year={1995}
}

@INPROCEEDINGS{PredicateAnalysis-1996MICRO,
  author={Johnson, R. and Schlansker, M.},
  booktitle={Proceedings of the 29th Annual IEEE/ACM International Symposium on Microarchitecture. MICRO 29}, 
  title={Analysis techniques for predicated code}, 
  year={1996},
  volume={},
  number={},
  pages={100-113},
  doi={10.1109/MICRO.1996.566454}}

@inproceedings{VLIWPredicate-1999ISCA,
  title={Value prediction in VLIW machines},
  author={Nakra, Tarun and Gupta, Rajiv and Soffa, Mary Lou},
  booktitle={Proceedings of the 26th annual international symposium on Computer architecture},
  pages={258--269},
  year={1999}
}

@article{AMDGPUPredicate-2011CompArchNews,
  title={Software-based branch predication for AMD GPUs},
  author={Taylor, Ryan and Li, Xiaoming},
  journal={ACM SIGARCH Computer Architecture News},
  volume={38},
  number={4},
  pages={66--72},
  year={2011},
  publisher={ACM New York, NY, USA}
}

@book{PredicatedExecution-1991Tech,
  title={On predicated execution},
  author={Park, Joseph CH and Schlansker, Mike},
  year={1991},
  publisher={Hewlett-Packard Laboratories Palo Alto, California}
}

@misc{NVIDIAGPUPredicate-2025online,
  title        = {Parallel Thread Execution ISA Version 9.0},
  author       = {NVIDIA},
  year         = 2025,
  url          = {https://docs.nvidia.com/cuda/parallel-thread-execution/},
  note         = {Accessed on Oct 25, 2025}
}

@INPROCEEDINGS{AdvancedPredication-2010FPT,
  author={Han, Kyuseung and Paek, Jong Kyung and Choi, Kiyoung},
  booktitle={2010 International Conference on Field-Programmable Technology}, 
  title={Acceleration of control flow on CGRA using advanced predicated execution}, 
  year={2010},
  volume={},
  number={},
  pages={429-432},
  doi={10.1109/FPT.2010.5681452}
}

@ARTICLE{TLIA-2016TPDS,
  author={Liu, Leibo and Wang, Junbin and Zhu, Jianfeng and Deng, Chenchen and Yin, Shouyi and Wei, Shaojun},
  journal={IEEE Transactions on Parallel and Distributed Systems}, 
  title={TLIA: Efficient Reconfigurable Architecture for Control-Intensive Kernels with Triggered-Long-Instructions}, 
  year={2016},
  volume={27},
  number={7},
  pages={2143-2154},
  doi={10.1109/TPDS.2015.2477841}}

@INPROCEEDINGS{TIA-2015ICCAD,
  author={Yin, Shouyi and Zhou, Pengcheng and Liu, Leibo and Wei, Shaojun},
  booktitle={2015 IEEE/ACM International Conference on Computer-Aided Design (ICCAD)}, 
  title={Acceleration of nested conditionals on CGRAs via trigger scheme}, 
  year={2015},
  volume={},
  number={},
  pages={597-604},
  doi={10.1109/ICCAD.2015.7372624}}

@INPROCEEDINGS{AURORA-2021DATE,
  author={Tan, Cheng and Xie, Chenhao and Li, Ang and Barker, Kevin J. and Tumeo, Antonino},
  booktitle={2021 Design, Automation \& Test in Europe Conference \& Exhibition (DATE)}, 
  title={AURORA: Automated Refinement of Coarse-Grained Reconfigurable Accelerators}, 
  year={2021},
  volume={},
  number={},
  pages={1388-1393},
  doi={10.23919/DATE51398.2021.9473955}}

@inproceedings{AcceleratorWall-HPCA2019,
  title={The accelerator wall: Limits of chip specialization},
  author={Fuchs, Adi and Wentzlaff, David},
  booktitle={2019 IEEE International Symposium on High Performance Computer Architecture (HPCA)},
  pages={1--14},
  year={2019},
  organization={IEEE}
}

@inproceedings{SSA-PPoPP1989,
  title={An efficient method of computing static single assignment form},
  author={Cytron, Ron and Ferrante, Jeanne and Rosen, Barry K and Wegman, Mark N and Zadeck, F Kenneth},
  booktitle={Proceedings of the 16th ACM SIGPLAN-SIGACT symposium on Principles of programming languages},
  pages={25--35},
  year={1989}
}

@misc{LLVMIfConversion-2026online,
  title        = {LLVM If-Conversion},
  author       = {LLVM},
  year         = 2026,
  url          = {https://llvm.org/doxygen/IfConversion_8cpp.html},
  note         = {Accessed on Mar 6, 2026}
}

@article{TRIPS-2004TACO,
author = {Sankaralingam, Karthikeyan and Nagarajan, Ramadass and Liu, Haiming and Kim, Changkyu and Huh, Jaehyuk and Ranganathan, Nitya and Burger, Doug and Keckler, Stephen W. and McDonald, Robert G. and Moore, Charles R.},
title = {TRIPS: A polymorphous architecture for exploiting ILP, TLP, and DLP},
year = {2004},
issue_date = {March 2004},
publisher = {Association for Computing Machinery},
address = {New York, NY, USA},
volume = {1},
number = {1},
issn = {1544-3566},
url = {https://doi.org/10.1145/980152.980156},
doi = {10.1145/980152.980156},
journal = {ACM Trans. Archit. Code Optim.},
month = mar,
pages = {62–93},
numpages = {32}
}

@inproceedings{Pipestitch-MICRO2023,
author = {Serafin, Nathan and Ghosh, Souradip and Desai, Harsh and Beckmann, Nathan and Lucia, Brandon},
title = {Pipestitch: An energy-minimal dataflow architecture with lightweight threads},
year = {2023},
isbn = {9798400703294},
publisher = {Association for Computing Machinery},
address = {New York, NY, USA},
url = {https://doi.org/10.1145/3613424.3614283},
doi = {10.1145/3613424.3614283},
booktitle = {Proceedings of the 56th Annual IEEE/ACM International Symposium on Microarchitecture},
pages = {1409–1422},
numpages = {14},
location = {Toronto, ON, Canada},
series = {MICRO '23}
}

@inproceedings{PICACHU-ASPLOS2025,
author = {Qin, Jiajun and Xia, Tianhua and Tan, Cheng and Zhang, Jeff and Zhang, Sai Qian},
title = {PICACHU: Plug-In CGRA Handling Upcoming Nonlinear Operations in LLMs},
year = {2025},
isbn = {9798400710797},
publisher = {Association for Computing Machinery},
address = {New York, NY, USA},
url = {https://doi.org/10.1145/3676641.3716013},
doi = {10.1145/3676641.3716013},
booktitle = {Proceedings of the 30th ACM International Conference on Architectural Support for Programming Languages and Operating Systems, Volume 2},
pages = {845–861},
numpages = {17},
location = {Rotterdam, Netherlands},
series = {ASPLOS '25}
}

@inbook{FexMo-MICRO2025,
author = {Yang, Yufei and Xie, Chenhao and Guo, Chuliang and Liu, Liansheng and Peng, Xiyuan and Liu, Datong and Peng, Yu},
title = {FexMo: Enabling Fuse Execution Mode for Multi-task CGRAs},
year = {2025},
isbn = {9798400715730},
publisher = {Association for Computing Machinery},
address = {New York, NY, USA},
url = {https://doi.org/10.1145/3725843.3756019},
booktitle = {Proceedings of the 58th IEEE/ACM International Symposium on Microarchitecture},
pages = {1236–1249},
numpages = {14}
}

@article{CHStone-IMT2009,
  title={Proposal and Quantitative Analysis of the CHStone Benchmark Program Suite for Practical C-based High-level Synthesis},
  author={Yuko Hara and Hiroyuki Tomiyama and Shinya Honda and Hiroaki Takada},
  journal={Information and Media Technologies},
  volume={4},
  number={4},
  pages={740-752},
  year={2009},
  doi={10.11185/imt.4.740}
}

@misc{PyTorch-Geomatric-axiv2019,
      title={Fast Graph Representation Learning with PyTorch Geometric}, 
      author={Matthias Fey and Jan Eric Lenssen},
      year={2019},
      eprint={1903.02428},
      archivePrefix={arXiv},
      primaryClass={cs.LG},
      url={https://arxiv.org/abs/1903.02428}, 
}
